\documentclass[12pt, draftclsnofoot, onecolumn,a4]{IEEEtran}
\usepackage{amsmath, mathrsfs, mathtools}
\usepackage{amsfonts}
\usepackage{amssymb}
\usepackage{color}
\usepackage{multirow}
\usepackage{stmaryrd}
\usepackage{yfonts}
\usepackage{mathabx}
\allowdisplaybreaks
\usepackage{caption}
\usepackage{subcaption}
\usepackage{float}
\usepackage{stfloats}
\usepackage{amsfonts}
\usepackage{amssymb}
\usepackage{color}
\usepackage{multirow}
\usepackage{graphicx}
\usepackage{tcolorbox}
\usepackage{mathrsfs}
\usepackage[numbers,sort&compress]{natbib}
\usepackage{multicol}
\usepackage{algorithm}
\usepackage{algpseudocode}
\usepackage{pifont}
\usepackage{varwidth}
\usepackage{float}
\usepackage{afterpage}
\usepackage{balance}
\allowdisplaybreaks
\usepackage{lipsum}
\usepackage{url}
\usepackage{mathbbol}

\def\mindex#1{\index{#1}}

\def\sq{\hbox{\rlap{$\sqcap$}$\sqcup$}}
\def\qed{\ifmmode\sq\else{\unskip\nobreak\hfil
\penalty50\hskip1em\null\nobreak\hfil\sq
\parfillskip=0pt\finalhyphendemerits=0\endgraf}\fi\medskip}

\long\def\defbox#1{\framebox[.9\hsize][c]{\parbox{.85\hsize}{%
\parindent=0pt
\baselineskip=12pt plus .1pt      
\parskip=6pt plus 1.5pt minus 1pt 
 #1}}}

\long\def\beginbox#1\endbox{\subsection*{}%
\hbox{\hspace{.05\hsize}\defbox{\medskip#1\bigskip}}%
\subsection*{}}

\def\endbox{}

\newsavebox{\junk}
\savebox{\junk}[1.6mm]{\hbox{$|\!|\!|$}}

\def\bC{{\mathbb C}}

\def\bE{{\mathbb E}}

\def\bR{{\mathbb R}}

\def\sfH{{\sf H}}

\def\bfmath#1{{\mathchoice{\mbox{\boldmath$#1$}}%
{\mbox{\boldmath$#1$}}%
{\mbox{\boldmath$\scriptstyle#1$}}%
{\mbox{\boldmath$\scriptscriptstyle#1$}}}}

\def\bfmY{\bfmath{Y}}

\def\bfmhhaY{\bfmath{\hhaY}} 
\def\bfmhhaY{\hbox to 0pt{$\widehat{\bfmY}$\hss}\widehat{\phantom{\raise 1.25pt\hbox{$\bfmY$}}}}

\def\til={{\widetilde =}}

\def\clC{{\cal C}}

\def\clN{{\cal N}}

 \def\FRAC#1#2#3{\genfrac{}{}{}{#1}{#2}{#3}}

\def\ddtp{{\mathchoice{\FRAC{1}{d^{\hbox to 2pt{\rm\tiny +\hss}}}{dt}}%
{\FRAC{1}{d^{\hbox to 2pt{\rm\tiny +\hss}}}{dt}}%
{\FRAC{3}{d^{\hbox to 2pt{\rm\tiny +\hss}}}{dt}}%
{\FRAC{3}{d^{\hbox to 2pt{\rm\tiny +\hss}}}{dt}}}}

\def\eqdef{\mathbin{:=}}

\def\average#1,#2,{{1\over #2} \sum_{#1}^{#2}}

\def\eye(#1){{\bf(#1)}\quad}

\newtheorem{theorem}{{\bf Theorem}}

\newtheorem{proposition}{{\bf Proposition}}

\newtheorem{remark}{{\bf Remark}}

\newtheorem{lemma}{{\bf Lemma}}

\def\eq#1/{(\ref{e:#1})}

\newcommand{\beqn}[1]{\notes{#1}%
\begin{eqnarray} \elabel{#1}}

\newcommand{\eeqn}{\end{eqnarray} }

\newcommand{\beq}[1]{\notes{#1}%
\begin{equation}\elabel{#1}}

\newcommand{\eeq}{\end{equation}}

\def\bdes{\begin{description}}
\def\edes{\end{description}}

\newcounter{rmnum}

\newcounter{anum}

{\end{list}}

\def\ass(#1:#2){(#1\ref{#1:#2})}

\def\ritem#1{
\item[{\sf \ass(\current_model:#1)}]
}

\newenvironment{recall-ass}[1]{%
\begin{description}
\def\current_model{#1}}{
\end{description}
}

\newcommand{\Sigmay}{{\Sigmay}_{\yv}}
\def\wh{\widehat}

\def\herm{{\sfH}}

\newcommand{\snrul}{{\sf snr}_{\rm ul}}
\newcommand{\snrdl}{{\sf snr}_{\rm dl}}
\newcommand{\hvt}{\hv}
\newcommand{\yvt}{\yv}

\allowdisplaybreaks

\def\cg{{\clC\clN}} 
\newcommand{\normd}[1]{{\left\vert\kern-0.25ex\left\vert\kern-0.25ex\left\vert #1 
		\right\vert\kern-0.25ex\right\vert\kern-0.25ex\right\vert}}
\long\def\comment#1{}

\newcommand{\av}{{\bf a}}
\newcommand{\bv}{{\bf b}}

\newcommand{\dv}{{\bf d}}

\newcommand{\fv}{{\bf f}}

\newcommand{\hv}{{\bf h}}

\newcommand{\uv}{{\bf u}}
\newcommand{\wv}{{\bf w}}
\newcommand{\vv}{{\bf v}}
\newcommand{\xv}{{\bf x}}
\newcommand{\yv}{{\bf y}}
\newcommand{\zv}{{\bf z}}

\newcommand{\Am}{{\bf A}}
\newcommand{\Bm}{{\bf B}}

\newcommand{\Fm}{{\bf F}}
\newcommand{\Gm}{{\bf G}}
\newcommand{\Hm}{{\bf H}}

\newcommand{\Rm}{{\bf R}}
\newcommand{\Sm}{{\bf S}}

\newcommand{\Um}{{\bf U}}
\newcommand{\Wm}{{\bf W}}
\newcommand{\Vm}{{\bf V}}
\newcommand{\Xm}{{\bf X}}

\newcommand{\Ac}{{\cal A}}

\newcommand{\Nc}{{\cal N}}

\newcommand{\Qc}{{\cal Q}}
\newcommand{\Rc}{{\cal R}}
\newcommand{\Sc}{{\cal S}}

\newcommand{\lambdav}{\hbox{\boldmath$\lambda$}}

\newcommand{\muv}{\hbox{\boldmath$\mu$}}

\newcommand{\phiv}{\hbox{\boldmath$\phi$}}
\newcommand{\psiv}{\hbox{\boldmath$\psi$}}
\newcommand{\thetav}{\hbox{\boldmath$\theta$}}

\newcommand{\Lambdam}{\hbox{\boldmath$\Lambda$}}

\newcommand{\Sigmam}{\hbox{\boldmath$\Sigma$}}
\newcommand{\Phim}{\hbox{\boldmath$\Phi$}}

\newcommand{\Psim}{\hbox{\boldmath$\Psi$}}

\newcommand{\trace}{{\rm Tr}}

\newcommand{\SNR}{{\sf snr}}

\newcommand{\transp}{{\sf T}}
\renewcommand{\vec}{{\rm vec}}

\newcommand{\hb}{\mathbb{h}}
\newcommand{\yb}{\mathbb{y}}

\title{FDD Massive MIMO Channel Training: \\ Optimal Rate-Distortion Bounds and the Efficiency of ``one-shot'' Schemes}
\author{Mahdi Barzegar Khalilsarai, Yi Song, Tianyu Yang, and Giuseppe Caire
\thanks{The authors are with the Communications and Information Theory Group (CommIT), Technische Universit\"{a}t Berlin, 10587 Berlin, Germany 
(e-mail: \{m.barzegarkhalilsarai, yi.song, tianyu.yang, caire\}@tu-berlin.de).}
\thanks{Part of this work has been presented in IEEE Int. Symp. on Inform. Theory (ISIT) 2022 \cite{khalilsarai2022channel_ISIT}.}
}

\begin{document}
	
	\maketitle
	
	\vspace{-1cm}
	
\begin{abstract}
We study the problem of providing channel state information (CSI) at the transmitter
in multi-user ``massive'' MIMO systems operating in frequency division duplexing (FDD). 
The wideband MIMO channel is a vector-valued random process correlated in time, space (antennas), and frequency (subcarriers).  
The base station (BS) broadcasts periodically $\beta_{\rm tr}$ 
pilot symbols from its $M$ antenna ports to $K$ single-antenna users (UEs). 
Correspondingly, the $K$ UEs send feedback messages about their channel state using $\beta_{\rm fb}$
symbols in the uplink (UL). 
Using results from remote rate-distortion theory, we show that, as $\SNR \to\infty$, 
the optimal feedback strategy achieves a channel state estimation mean squared error (MSE) that behaves as 
$\Theta (1)$ if $\beta_{\rm tr} < r$  and as $\Theta\left(\SNR^{-\alpha}\right)$ when $\beta_{\rm tr} \ge r$, where $\alpha = \min (\beta_{\rm fb}/r , 1)$, where $r$ is the rank of the channel covariance matrix. 
The MSE-optimal rate-distortion strategy implies encoding of long sequences of channel states, which would yield completely stale  CSI and therefore poor multiuser precoding performance. 
Hence,  we consider three practical ``one-shot'' CSI strategies  with minimum one-slot delay 
and analyze their large-SNR channel estimation MSE behavior. 
These are: (1) digital feedback via entropy-coded scalar quantization (ECSQ), 
(2)  analog feedback (AF), and (3) local channel estimation at the UEs and digital 
feedback. These schemes have different requirements in terms of knowledge of the channel statistics at the UE and at the BS. In particular, the latter strategy requires no statistical knowledge and is closely inspired by a CSI feedback scheme 
currently proposed in 3GPP standardization.
It is shown that ECSQ achieves optimal MSE at the price of a slight increase in feedback rate which vanishes for large SNR. 
AF achieves the optimal MSE decay rate of $\Theta (\SNR^{-1})$ whenever $\beta_{\rm tr},\beta_{\rm fb} \ge r$ but is sub-optimal if $\beta \ge r$ and $\beta_{\rm fb} < r$. 
The {\em 3GPP-inspired} scheme is shown, via numerical simulations, to achieves performance similar to
ECSQ and AF when the multipath channel is sufficiently sparse in the angle-delay domain, 
but suffers from a large performance gap if this requirement is not met.
\end{abstract}

	
	\section{Introduction}  \label{sec:intro}
	
	Multiuser (massive) MIMO consists of serving $K > 1$ users (UEs) on the same time-frequency resource dimension 
	using a larger number of antennas ($M\gg 1$) at the Base Station (BS) via spatial multiplexing \cite{marzetta2010noncooperative,marzetta2016fundamentals}. 
	Achieving the remarkable benefits of massive MIMO relies at large on the availability of accurate channel state information (CSI) at the BS. 
	In particular, the downlink (DL) requires that the BS computes precoding vectors as a function of the users' DL CSI. In time division duplexing (TDD) 
	systems, DL CSI is obtained from uplink (UL) pilots through channel reciprocity. In frequency division duplexing (FDD), where reciprocity does not hold, the BS needs to train user channels by broadcasting training pilots in DL and receiving CSI 
	 feedback in UL. 
	 Training with finite-rate feedback results in imperfect CSI estimates, which translates into a loss in DL spectral efficiency, the proportions of which depends on the CSI estimate quality. It is for instance well-known that when the mean squared error (MSE) between true and estimated CSI of the users decreases as $O(\SNR^{-1})$, then zero-forcing (ZF) beamforming achieves the full system Degrees of Freedom (DoF) despite imperfect CSI \cite{jindal2006mimo}. 
	 It is also known that if the error decreases as $O(\SNR^{-\alpha})$ for some $\alpha \in [0,1]$ for all $K$ users, then the optimal achievable DoF 
	 is given by $1+(K-1)\alpha$, and these are achievable via the rate-splitting method \cite{joudeh2016sum}.\footnote{The system DoF is defined as the limit of the total DL spectral efficiency (in bit per time-frequency symbol) divided by $\log \SNR$, as $\SNR \rightarrow \infty$. This is also known as the ``pre-log'' factor of the sum spectral efficiency, or also as the system total ``multiplexing gain'' \cite{joudeh2016sum}.} 
	 
	 The parameter $\alpha$ is known in the literature as the \textit{quality scaling exponent} (QSE) \cite{jindal2006mimo,yang2012degrees}. In this paper, we present an upper-bound on the QSE achieved by \textit{any} DL training and UL feedback strategy in FDD systems 
	under the following assumptions. We consider a BS with $M$ antennas and OFDM data transmission format over $N$ subcarriers. We 
	model the channel as a correlated Gaussian vector process that evolves in time according to a block-fading model, in which the channel is constant, and independently realized, over coherence time intervals of duration $T$ OFDM symbols (in time).\footnote{Each OFDM symbol corresponds to $N$ signal samples in the time domain plus the cyclic prefix, which in the context of this paper is ignored since it is irrelevant. This means that an OFDM symbol spans $N$ time-frequency signal dimensions, 
	in the OFDM time-frequency frame.} 
	The channel statistics, namely mean and covariance matrix of the channel vectors, 
	are assumed constant over time intervals much longer than the coherence time. 
	The BS trains the channels of all users by broadcasting $\beta_{\rm tr}$ 
pilot symbols over each coherence block of $\beta = TN$ time-frequency symbols.
Upon receiving noisy pilots, the users compute their respective feedback messages and send them to the BS in the UL using 
$\beta_{\rm fb}\le \beta$ symbols. To simplify matters, we assume that DL pilot transmission and UL feedback take place in the same channel coherence block, so that there is no ``channel aging'' \cite{truong2013effects}.\footnote{The effect of channel aging can be taken into account by an appropriate channel
predictor. However, this goes beyond the scope of this paper. Here we prefer to focus uniquely on the performance of DL training and UL feedback.}  

Under these assumptions, we show that the QSE of any feedback strategy is upper-bounded by a value that depends on three parameters: the channel covariance rank $r$, the training dimension (or pilot length) $\beta_{\rm tr}$, and the feedback dimension $\beta_{\rm fb}$. The bound is derived by modelling the user as an encoder that 
observes outputs of a source (the channel vector) through noisy linear measurements (DL pilot symbols). 
This is a particular instance of the general  {\em remote source coding} problem in rate-distortion theory \cite{berger1971rate}. The encoded source is then transmitted over a channel 
of given capacity, which yields a channel estimate at the BS. We show that, for given $r,\, \beta_{\rm tr},\, \beta_{\rm fb}$, as SNR$\to\infty$, the MSE of the estimated CSI at the BS behaves at best as $\Theta(\SNR^{-\alpha_{\rm rd}})$ where $\alpha_{\rm rd} = 0$ whenever $\beta_{\rm tr}<r$ (insufficient training) and $\alpha_{\rm rd} = \min (\beta_{\rm fb}/r,1)$ when $\beta_{\rm tr} \geq r$. More specifically, if $\beta_{\rm tr}\ge r$ and $\beta_{\rm fb}<r$ (sufficient training, insufficient feedback), the error decays as a function of SNR with a less-than-one fractional exponent $\beta_{\rm fb}/r$ hence achieving only a fraction of the full 
DoF, and if $\beta_{\rm tr},\beta_{\rm fb}\ge r$ (sufficient training and feedback), the error decays with an exponent of $1$ and achieves the full DoF.

	The feedback strategy that achieves the optimal rate-distortion trade-off involves knowledge of the channel statistics both at the BS and the UE  side, 
	as well as employment of high-dimensional vector quantizers operating on infinite-dimensional sequences of channel observations. 
	While achieving the best CSI MSE, this strategy is totally impractical since it incurs a large feedback delay, thus providing completely stale CSI, which is therefore useless for computing the multi-user MIMO precoding vectors on each coherence block. 
Hence, we consider the following three alternatives for practical  one-shot schemes where the feedback message is a function of the instantaneous channel estimate in the current slot only. 
\begin{enumerate}
	\item \textbf{ECSQ Feedback:} We first consider a scheme where the UE has statistical knowledge of its DL channel, it quantizes the Karhunen-Lo{\`e}ve (KL) expansion 
	coefficients via dithered scalar quantization, applies entropy coding for feedback compression  \cite{ziv1985universal}, 
	and sends the {\em digital} feedback message  consisting of the entropy-coded bits.
This scheme is referred to as entropy-coded scalar quantization (ECSQ). We show that ECSQ achieves optimal distortion for all SNR values with an overhead 
in terms of the feedback symbols with respect to the rate-distortion strategy which vanishes as SNR$\to \infty$ (hence yielding a QSE of $\alpha_{\rm ecsq} = \alpha_{\rm rd}$). 

\item \textbf{Analog Feedback:} The next strategy is known as Analog Feedback (AF) in the literature \cite{marzetta2006fast,caire2010multiuser,kobayashi2011training}, 
which lifts the requirement of channel statistics knowledge at the UE and simply sends the DL pilot slot of length $\beta_{\rm tr}$ in UL by ``spreading" it 
over $\beta_{\rm fb}$ dimensions using unquantized  quadrature amplitude modulation (QAM).\footnote{Interestingly, a similar
direct transmission of unquantized source symbols with spreading is proposed and used by Amimon \url{https://www.amimon.com/} for joint source channel coding of video, thus demonstrating the practical feasibility of such ``analog feedback'' technique.} 
The BS then computes a minimum mean squared error (MMSE) estimate of the channel given the feedback. 
We show that with sufficient training and feedback ($\beta_{\rm tr},\beta_{\rm fb}\ge r$), AF achieves the optimal QSE  ($\alpha_{\rm af} = \alpha_{\rm rd}$), while it is strictly sub-optimal when feedback is insufficient $(\beta_{\rm fb}<r)$ and yields a constant residual error even as SNR$\to\infty$ ($\alpha_{\rm af}=0$ vs $\alpha_{\rm rd} = \beta_{\rm fb}/r$). \item \textbf{CS-Based Feedback:} Finally, motivated by the scheme currently proposed in 3GPP, 
we consider a scheme that requires no channel statistical knowledge (either at the BS or at the users). 
In this {\em 3GPP-inspired} scheme, the UE  estimates the channel coefficients in the angle-delay domain 
and feeds back the quantization bits relative to the dominant coefficients. 
Since the DL training dimension $\beta_{\rm tr}$ is generally smaller than the ambient dimension of the channel $NM$ 
(subcarriers $ \times$ antennas), the channel estimation at the UE must use some form of {\em compressed sensing} (CS), 
which in turns relies on channel sparsity in the angle-delay domain. Notice that whenever $\beta_{\rm tr} < NM$ the channel estimation problem is intrinsically a CS problem, 
and this corresponds to the widely studied case of ``compressed'' DL training
(e.g., see \cite{shen2016compressed,han2017compressed,zhang2018distributed,ding2015channel}).  
We show that the {\em 3GPP-inspired} scheme incurs a constant residual error as SNR$\to \infty$ and hence a QSE of zero in all non-trivial scenarios ($\alpha_{\rm cs}= 0$), although it can be quite effective in moderate SNR values and under favorable channel sparsity conditions.
\end{enumerate}
	
\subsection{Related Works}
	
The analysis of the CSI estimation error and its effect on DL spectral efficiency in FDD MIMO systems has been the subject of plenty of works in the literature (e.g., see \cite{jindal2006mimo,ding2007multiple,caire2010multiuser,shirani2009channel}). 
In \cite{caire2010multiuser,shirani2009channel} upper-bounds on the difference between the achievable rates with ZF beamforming with perfect and estimated CSI are provided, with comparisons between AF and digital (codebook based) feedback schemes. 
	These works consider a ``narrowband'' MIMO channel ($N=1$ subcarrier, or equivalently, a frequency non-selective channel) and assume
$\beta_{\rm tr}\ge M$ DL pilots, which is reasonable for small and moderate number of antennas $M$, but unrealistic when considering massive MIMO. 
We depart from this assumption by considering the training dimension to be an arbitrary value, possibly less than the number of BS antennas. We also show that, different from the conclusion of \cite{jindal2006mimo} and in line with the conclusions of 
\cite{caire2010multiuser,kobayashi2011training}, taking into account the estimation error at the user, analog feedback achieves the same error decay rate of $\Theta (\SNR^{-1})$ as digital feedback (operating at the rate-distortion bound), with no knowledge of channel statistics at the user side and with very minimal signal processing, provided that the number of training and feedback symbols ($\beta_{\rm tr}$ and $\beta_{\rm fb}$) are no less than the channel covariance rank $r$.

On the algorithmic side, several works have considered designing DL training pilots and feedback codebooks for massive MIMO channels that achieve certain performance criteria \cite{jiang2015achievable,bazzi2018amount,gu2019information}. However, these works do not suggest a definite answer to the question of how much training and feedback is needed to achieve a target error and/or spectral efficiency performance. One reason is that the design of pilot matrices and feedback codebooks in these works depends on the particular statistics of the user channels. In contrast, we consider a rather generic design of training pilots as independently generated Gaussian symbols, and we characterize the large-SNR behavior of the channel estimation error in terms of training and feedback dimensions. This characterization depends on the channel statistics only in terms of the covariance rank.   
	
As said before, many works have focused on CSI estimation via CS methods \cite{shen2016compressed,han2017compressed,zhang2018distributed,ding2015channel}. 
The underlying assumption in all these works is that the channel vector can be approximated as $\hv \approx \Phim \wv$, for some known dictionary matrix  $\Phim$ and $s$-sparse, unknown $\wv$ where $s$ is much smaller than the dimension of $\hv$. A plethora of CS algorithms have been devised to estimate $\hv$ from its $m$ noisy linear projections, and it is well-known that stable reconstruction\footnote{In the literature of compressed sensing, stable reconstruction means that 
the MSE in estimating $\hv$ from noisy linear projections in the form 
$\yv = \Xm \hv + \zv$ is bounded by a constant (independent of the dimension of $\hv$) times the variance of the additive noise $\zv$. This also implies that the MSE vanishes as the SNR in the observation grows to infinity.} 
is possible if $m$ is larger than $s$ times a logarithmic term in the dimension of $\hv$ 
(e.g., see \cite{reeves2008sampling}). This allows to transmit a number of DL pilot symbols that depends on the channel sparsity order rather than
on the channel dimension, providing large potential savings in the DL pilot overhead.
In this work, we show that while CS-based training and feedback can be very beneficial in a range of moderate SNR, 
it involves a non-vanishing CSI error and therefore a null QSE. 
We emphasize this point by considering a well-known CS-based channel estimation technique 
in the {\em 3GPP-inspired} scheme considered in this paper. 

	\section{Downlink Channel Training}\label{sec:ch_training}
	
In the considered system, a  BS with $M$ antennas serves $K\le M$ single-antenna user equipments (UEs), 
operating in FDD mode with OFDM data transmission over $N$ subcarriers. 
In this section we omit the user index since the DL training is common (broadcasted to all users) and each user operates its feedback scheme independently. 
The frequency-domain symbol over subcarrier $n$, received by a generic UE, is given by 
	\[y[n] =\hvt^\herm [n]  \xv [n] + z[n],\]
	where $\hvt [n]  \in \bC^M$ consists of the channel fading coefficients between the $M$ BS antennas and that of the UE at subcarrier $n$, $\xv[n]  \in \bC^M$ is the vector of transmitted symbols from the $M$ array antennas satisfying the transmission power constraint $\bE [\Vert \xv [n]\Vert^2]\le \snrdl $ for all $n$, where $\snrdl$ denotes signal-to-noise ratio (SNR) in the DL, while $z\sim \cg (0,1)$ is zero-mean, additive white Gaussian noise (AWGN).\footnote{We have simplified the notation by normalizing the additive noise variance and accordingly defining the BS transmission power as equivalent to the SNR. We note that this gives the so-called ``pre-beamforming" SNR, which is different from defining the SNR as the ratio between the expected signal power and the expected noise power at the receiver side.} Concatenating the channel over all subcarriers, we define the wideband channel as the vector $\hv \eqdef [ \hvt[1]^\transp,\ldots, \hvt[N]^\transp ]^\transp \in \bC^{MN}$, which we refer to as the channel state information at the transmitter (CSIT). We assume that $\hv$ evolves according to a block-fading model where it is constant over frames of duration $T$ OFDM symbols
	(i.e., blocks of $\beta = T N$ time-frequency symbols) 
	and changes frame to frame according to an i.i.d. process, namely $\hv \sim \cg (\muv,\Sigmam^h)$ where $\muv = \bE[\hv]$ and $\Sigmam^{h} = \bE \left[\left(\hv - \muv \right)  \left( \hv-\muv \right)^{\herm} \right]$ is the channel covariance of rank $r = {\rm rank}(\Sigmam^h)$. Notice that the channel correlation in space (antenna) and frequency (subcarriers) is completely characterized by $\Sigmam^h$.
	
	To train the DL channel, the BS broadcasts pilot symbols in each frame of dimension $\beta$. The pilot symbols are placed in the time-frequency grid 
	over $T_p \le T$ OFDM symbols in time and  over a subset of $N_p$ subcarriers in frequency, for a total of $\beta_{\rm tr} = T_p N_p$ pilot symbols per frame. 
	We denote the pilot subcarriers by the set of indices $\Nc_p \subseteq \{1,\ldots,N\}$. The training measurements received at the UE can be written as
	\begin{equation}\label{eq:subcar_measurements}
		\yvt^{\rm tr}[n] = \hvt^\herm [n] \Xm^{\rm tr}  [n]+ \zv^{\rm tr} [n],~n\in \Nc_p,
	\end{equation}
	where $\Xm^{\rm tr}[n] = \left[\xv^{\rm tr}_1[n],\ldots,\xv^{\rm tr}_{T_p}[n]\right]\in \bC^{M\times T_p}$ is a matrix, where the $t$-th column $\xv^{\rm tr}_t[n]$ denotes the vector of $M$ training symbols sent at subcarrier $n$ at time instant $t$, 
	and where row $m$ represents the length-$T_p$ pilot sequence transmitted from antenna $m$ at subcarrier $n$. 
	Collecting the $\beta_{\rm tr}$ training measurements over all pilot subcarriers, the received pilot signal at the UE side $\yv^{\rm tr} = [\yvt [n_1], \ldots,\yvt [n_{N_p}]]\in \bC^{1\times \beta_{\rm tr}}$ can be written as
	\begin{equation}\label{eq:training_signal}
		\yv^{\rm tr} =\hv^\herm  \Xm^{\rm tr} + \zv^{\rm tr},
	\end{equation}
	where $\Xm^{\rm tr} = \left( \Bm_{n,\ell} \right)_{n\in [N]}^{\ell\in [N_p]}\in \bC^{MN\times \beta_{\rm tr}}$ is the \textit{training matrix} consisting of $M\times T_p$ blocks $\Bm_{n,\ell}$, where $\Bm_{n,\ell} = \Xm^{\rm tr}[n_{\ell}]$ if $n_\ell \in \Nc_p$ and $\Bm_{n,\ell} = \mathbf{0}$ if $n_\ell \notin \Nc_p$. 
For the sake of analytical tractability, in this paper 
we assume the $M$-dimensional pilot vectors to be generated according to 
	\begin{equation}\label{eq:isotropic_pilots}
		\xv_t^{\rm tr}[n] \sim \cg \left (\mathbf{0},\frac{\snrdl}{M}\mathbf{I} \right ),~t\in \{1,\ldots, T_p\}, \, n\in \Nc_p.
	\end{equation}
It can be shown that this choice suffices to achieve optimal QSE. 	
Upon receiving the pilots, the UE computes a feedback message  containing information about the channel state 
and sends it to the BS over $\beta_{\rm fb}$ uses (i.e., time-frequency symbols) of the UL channel. 


We model the UL as a MIMO multiple access channel (MIMO-MAC) where all the UEs send their feedback simultaneously to the BS. We assume for simplicity that the BS has perfect knowledge of the UL channels. Such a knowledge can be obtained via UL pilots embedded in the UL frame. We denote the peak transmission power for each user by $\snrul$, which we assume to be a constant factor of the BS transmission power, i.e. $\snrul = \kappa \snrdl$, where typically $\kappa<1$ since the users can transmit with less power than the BS. Similar to \cite{jiang2015achievable}, we consider the MIMO-MAC uplink channel capacity formula $C^{\rm ul} = \log (M \snrul) = \log (M \kappa \snrdl) $ (per user) with a diversity-multiplexing trade-off factor of one \cite{caire2010multiuser}.  Notice that in order to compare the efficiency of AF with other digital feedback techniques it is of fundamental importance to 
convert ``bits'' into UL ``signal dimensions'', which in turns requires some notion of spectral efficiency of the UL channel \cite{kobayashi2011training}.
The simple UL capacity formula used here allows for a more elegant development 
of the theory without blurring the main results. However, this assumption is not fundamental, in the sense that one can obtain similar results by replacing other feedback channel models, as long as the high-SNR capacity of the feedback channel grows with $\log (\snrdl)$.

Finally, we consider the mean squared error (MSE),
	\begin{equation}\label{eq:err_def}
		d(\hv,\widehat{\hv}) = \bE \left[ \Vert \hv - \widehat{\hv} \Vert^2 \right],
	\end{equation}
as the distortion metric between true $\hv$ and estimated CSIT $\widehat{\hv}$.
	From the estimated CSIT, the BS computes an $M\times K$ precoding matrix $\Vm [n]$ for each subcarrier $n$. Here we consider the standard and commonly used ``naive'' zero-forcing (ZF) precoding,\footnote{This refers to the fact that the precoder naively treats the estimated channels as true channels.} where 
	$\Vm [n]$ is the column-normalized pseudo-inverse of the estimated channel matrix $\widehat{\Hm} [n] =[ \widehat{\hv}_1[n],\ldots, \widehat{\hv}_K[n] ]$, where $\widehat{\hv}_k[n]$ denotes the estimated channel of user $k$ at subcarrier $n$. The channel estimation QSE is  defined as \cite{joudeh2016sum,jindal2006mimo}
	\begin{equation}\label{eq:QSE_def}
		\alpha = \lim_{\snrdl\to \infty }-\frac{\log\, d(\hv,\widehat{\hv})}{\log \, \snrdl},
	\end{equation}
	indicating the rate of error decay in large SNR. We characterize the large-SNR error behavior slightly differently, using the $\Theta$ notation as $d(\hv , \widehat{\hv})= \Theta (\snrdl^{-\alpha})$ which implies a QSE of $\alpha$. \footnote{We use the standard Bachmann-Landau  {\em order notation} 
	for which $f(x) = O(g(x))$ if there exists constants $C_0$ and $x_0$ for which $|f(x)| \leq C_0 g(x)$ for all $x \geq x_0$, 
	$f(x) = \Omega(g(x))$ if $g(x) = O(f(x))$ and $f(x) = \Theta(g(x))$ if both $f(x) = O(g(x))$ and $f(x) = \Omega(g(x))$.}

	
	\section{Lower bound on CSIT estimation error}\label{sec:RD_LB}
	
	We can think of the feedback message as a code that contains information about the channel, 
	computed at the user side. The instantaneous channel is a random realization of a \textit{source} with distribution $\cg (\muv , \Sigmam^h)$. For each source realization, 
the user (encoder) observes $\beta_{\rm tr}$ noisy linear measurements given by (\ref{eq:training_signal}). 
Since the encoder does not have direct access to the source, the problem is an example of \textit{remote source coding} \cite{berger1971rate}. 
In this context, the source is a stationary sequence of i.i.d. vector symbols 
$\hb^{S} = \{ \hv^{(i)} \sim \cg (\muv,\Sigmam^h) \}_{i=1}^{S} $, where $\hv^{(i)}$ denotes the channel in frame $i$ and $S$ is the source-encoding block length.\footnote{The superscript $i$ is included here to distinguish between channel realizations across different frames, but is omitted elsewhere when such distinction is not necessary.} 
The per-frame pilot transmission in DL yields a sequence of measurements $\yb^{S}=\{ \yv^{{\rm tr}\, (i)} \}_{i=1}^{S}$ such that, from \eqref{eq:training_signal}, 
we have $\yv^{{\rm tr}\, (i)} = \hv^{(i)\, \herm}\Xm^{\rm tr} + \zv^{{\rm tr}\, (i)}$. 
A rate $R$
	remote rate-distortion code of block length $S$ consists of an encoding function 
	$ f_{S}\,: \yb^{S}\to \{1,\ldots,2^{S R}\},$ and a decoding function $g_{S}\,: \{1,\ldots,2^{S R}\}\to \widehat{\hb}^S $, 
where $\widehat{\hb}^{S} = \{ \widehat{\hv}^{(i)}\}_{i=1}^{S}$ is the sequence of channel estimates. 
A remote rate-distortion pair $(R,D)$ is said to be achievable if there exists a sequence of rate $R$ codes such that
	\begin{equation}\label{eq:d_avg}
		d_{\rm avg} \eqdef \underset{S \to \infty}{\lim}  \frac{1}{S} \sum_{i=1}^{S} d(\hv^{(i)},\widehat{\hv}^{(i)})\le D
	\end{equation}
	where $d(\cdot,\cdot)$ is defined in \eqref{eq:err_def}.
	 The remote rate-distortion function $R_{\hv}^r(D)$ is the infimum of rates $R$ such that $(R,D)$ is achievable.
	 The inverse function $D_{\hv}^r (R)$, i.e., the infimum of $D$ such that $(R,D)$ is achievable, is referred to as the remote distortion-rate function \cite{cover2006elements}.
	
	Let us denote the linear MMSE estimator of the channel given the pilots by 
	\begin{equation}
		\begin{aligned}
			\uv^{(i)} &= \bE[\hv^{(i)}|\yv^{{\rm tr}\, (i)} ] = \Sigmam^h \Xm^{\rm tr}\left(  \Xm^{{\rm tr}\, \herm} \Sigmam^h \Xm^{{\rm tr}} + \mathbf{I} \right)^{-1} (\yv^{{\rm tr}\, (i) } - \muv^\herm \Xm^{\rm tr} )^\herm + \muv,   \label{MMSE-u}
		\end{aligned}
	\end{equation}
	for $i=1,\ldots,S$. For a given realization of $\Xm^{\rm tr}$, the $\uv^{(i)}$ are independent realizations of a Gaussian vector-valued i.i.d. process with mean $\muv$ and covariance 
	\begin{equation}\label{eq:Sigma_u}
		\Sigmam^u = \Sigmam^h \Xm^{\rm tr}\left(  \Xm^{{\rm tr}\, \herm} \Sigmam^h \Xm^{{\rm tr}} + \mathbf{I} \right)^{-1}\Xm^{{\rm tr}\, \herm} \Sigmam^h. 
	\end{equation}
	We denote the non-zero eigenvalues of $\Sigmam^{u}$ by the ordered set $\{\lambda_{1}^{u} \ge \ldots \ge  \lambda_{\widetilde{r}}^{u} \}$, where $\widetilde{r} \le \min (r , \beta_{\rm tr})$ is the rank of $\Sigmam^u$. The following lemma gives an expression for the remote rate-distortion function $R_{\hv}^r(D)$.
	\begin{lemma}[remote rate-distortion function]\label{lem:rate_distortion}
		For a fixed realization of the training matrix $\Xm^{\rm tr}$, the remote rate-distortion function is given by~\footnote{$[\cdot]_+$ indicate the positive part, i.e., $[x]_+ = \max\{x,0\}$.}
		\begin{equation}\label{eq:rrd_expression}
			R_{\hv}^r (D) = \sum_{\ell=1}^{\widetilde{r}}\left[ \log \frac{\lambda_{\ell}^{u}}{\gamma} \right]_+,~\text{for}~ D\ge D_{\rm mmse}
		\end{equation}
		where $\gamma$ is chosen such that $\sum_{\ell=1}^{\widetilde{r}}\min \{ \gamma, \lambda_{\ell}^{u} \} = D-D_{\rm mmse},$ and where $D_{\rm mmse} = \bE [\Vert \hv^{(i)} -\uv^{(i)}  \Vert^2] ={\rm Tr} \left( \Sigmam^h - \Sigmam^u \right)$ is the MMSE, where ${\rm Tr}(\cdot)$ denotes the trace. Similarly, the distortion-rate function is given by 
		\begin{equation}\label{eq:rdr_expression}
			D_{\hv}^r (R) =D_{\rm mmse} + \sum_{\ell=1}^{\widetilde{r}}\min \{ \gamma', \lambda_{\ell}^{u} \},~\text{for}~R\ge 0
		\end{equation}
		where $\gamma'$ is chosen such that $R = \sum_{\ell=1}^{\widetilde{r}}\left[ \log \frac{\lambda_{\ell}^{u}}{\gamma'} \right]_+$. 
	\end{lemma}	
\begin{proof}
See Appendix \ref{app:RD_lemma_proof}.
	\end{proof}

	With the rate-distortion feedback, the user encodes the source $\hb^{S}$ with $R_{\hv}^r(D)$ bits per sample on average and sends the quantization index in the UL via a channel code, after which the BS computes a sequence of channel estimates $\widehat{\hb}^S$. A direct application of the source-channel separation theorem (see \cite{cover2006elements}, exercise 10.17) yields that if the feedback channel has capacity $C^{\rm ul}$ bits per (time-frequency) symbol 
and if we use $\beta_{\rm fb}$ symbols per source vector to send the feedback message, we can achieve an MSE of $D$ in the sense of 
\eqref{eq:d_avg} (i.e., $d_{\rm avg}\le D$) if and only if $\beta_{\rm fb} C^{\rm ul} > R_{\hv}^r(D)$. 
Therefore, for a given feedback dimension $\beta_{\rm fb}$, no feedback strategy can achieve a channel estimation MSE smaller than $D_{\hv}^r ( \beta_{\rm fb} C^{\rm ul} )$. This MSE lower-bound depends on the training and feedback dimensions ($\beta_{\rm tr}$ and $\beta_{\rm fb}$). The dependence on $\beta_{\rm tr}$ is implicit in the following sense: the distortion-rate function depends on the eigenvalues of the covariance $\Sigmam^u$ given in \eqref{eq:Sigma_u}, which in turns is a function of the channel covariance as well as the training matrix with dimension $MN\times \beta_{\rm tr}$. In the general case it is very difficult to make the dependence of the distortion-rate function on the training dimension $\beta_{\rm tr}$ more explicit. However, we characterize this dependence in the large SNR regime via the following theorem.
	\begin{theorem}\label{thm:rate_asymp}
		For channels with covariance rank $r = {\rm rank}(\Sigmam^h)$, the rate-distortion feedback strategy achieves a channel estimation MSE that behaves as $\Theta (\snrdl^{-\alpha_{\rm rd}})$ with probability one over the realizations of the training matrix $\Xm^{\rm tr}$, where the QSE $\alpha_{\rm rd}$ is given by~\footnote{$\mathbf{1}\{ \Ac\}$ denotes the indicator function of the condition $\Ac$, i.e., it returns $1$ wher $\Ac$ is true, and 0 when $\Ac$ is false.}
		\begin{equation}\label{eq:decay_exponent_ub}
			\alpha_{\rm rd}= \min (\beta_{\rm fb}/r,1)\mathbf{1}\{\beta_{\rm tr} \geq r \}.
		\end{equation}	
	\end{theorem}
	\begin{proof}
		See Appendix \ref{app:rate_asymp_thm_proof}.
	\end{proof}
%

%
	\section{Practical Feedback Schemes}\label{sec:practical_feedback}
	
The rate-distortion bound of Section \ref{sec:RD_LB} is achieved by encoding infinite-dimensional blocks of MMSE channel estimates $\{\uv^{(i)} \}_{i=1}^{S},\, S\to \infty$ given in \eqref{MMSE-u}. As already remarked in Section \ref{sec:intro}, this would incur an impractical feedback delay and, 
	although minimizing the channel
	estimation MSE at the BS, it would yield {\em stale} channel estimates that cannot be used for DL precoding.  
	In this section we examine practical one-shot estimation and feedback schemes where the feedback message at each frame $i$ is a function of the 
	received DL pilot  signal $\yv^{{\rm tr} \,(i)}$ in the current frame only. Interestingly, as we will show, such schemes achieve CSIT estimation errors similar to that of the optimal rate-distortion scheme at the cost of a tolerable increase in the feedback rate. 
	
	\subsection{Entropy-Coded Scalar Quantization}
	
	The feedback method based on entropy-coded scalar quantization (ECSQ) assumes channel covariance knowledge at both the BS and the UE side and computes the feedback message as follows. In any frame, the UE computes the MMSE channel estimate $\uv = \bE [\hv | \yv^{\rm tr}]$ from the pilot signal $\yv^{\rm tr}$, where we have dropped the superscript $i$ from the variables for simplicity. The Karhunen-Lo{\`e}ve (KL) expansion of $\uv$ is given by 
	\begin{equation}\label{eq:KL_exp}
		\uv =  \muv + \Fm \wv, 
	\end{equation}
	where $\Fm = [\fv_1,\ldots,\fv_{MN}]$ contains the (orthonormal) eigenvectors $\{\fv_{\ell} \}_{\ell=1}^{MN}$ of $\Sigmam^u$ as its columns, $\wv = [w_1,\ldots,w_{MN}]^\transp$, with  $w_{\ell} \sim \cg(0,\lambda_{ \ell}^u)$,  are the KL 
coefficients, and $\{\lambda_{ \ell}^u\}$ are the eigenvalues of $\Sigmam^u$. 
Since $\muv$ and $\Fm$ 	 are known and $\Fm$ is unitary,  
the MSE incurred by quantizing the KL coefficients $\wv$ is the same as that of quantizing $\uv$, i.e. $\bE [\Vert \wv - \widehat{\wv}\Vert^2 ] =  \bE [\Vert \uv - \widehat{\uv}\Vert^2 ]$.
	
We follow the classical result of \cite{ziv1985universal} on entropy-coded dithered scalar quantization.
Let $\Sc = \{ \ell \, : \, \lambda_{ \ell}^u > 0 \}$ of size $\widetilde{r} = |\Sc|$ denote the set of KL coefficients with positive variance, and denote by $\wv_{\Sc}$ the vector of such coefficients.
Consider a uniform scalar quantizer $Q_1: \bR \to \bR$ with quantization points $C_1 = \{ 0 , \pm \Delta,\pm 2\Delta,\ldots \}$. 
Define the dithering random variable $Z_q$  statistically independent of $\wv$ and uniformly distributed over the interval $[-\Delta/2,\Delta/2]$. 
This dithering variable is known to both the BS and the UE.\footnote{Notice that the dithering variable is 
common randomness, which can be approached in practice by suitably synchronized 
pseudo-random number generation, analogous to random-spreading CDMA, frequency hopping, and many other schemes that require 
some form of  common randomness. Furthermore,  the dithering is needed to obtain a tractable rate-distortion distortion bound, but in practice it is well-known that a slightly better performance can be achieved without dithering. Therefore, the assumption of common randomness is not at all a limiting factor in this scheme.}
Define a vector $\zv_q$ of dimension $\widetilde{r}$ as $\zv_q = [Z_q+jZ_q,\ldots,Z_q+jZ_q]$. The UE represents the vector of coefficients $\wv$ with a code point $\widehat{\wv}$ whose elements on the index set $\Sc$ are given by $\widehat{\wv}_{\Sc} = Q_1(\wv_{\Sc}+\zv_q)-\zv_q$, where $Q_1$ is applied element-wise to the real and imaginary parts of its vector input, 
and the rest of its elements are set to $0$. 
From Lemma 1 of \cite{ziv1985universal} we have that this scheme yields an MSE satisfying the bound
	\begin{equation}\label{eq:dither_err}
		\bE \left[ \Vert \widehat{\wv} - \wv \Vert^2 \right] \leq \widetilde{r}\Delta^2 / 6.
	\end{equation}
	In fact, this distortion is {\em universal}, in that it is independent of the distribution of $\wv$. 
	It follows that the distortion $D$ in estimating the channel is achievable by choosing the step size 
	$\Delta =\sqrt{6D/\widetilde{r}}$. It is also shown in \cite{ziv1985universal} that for the target distortion value of $D$, the 
	above dithered scalar quantizer  followed by 
	lossless \textit{entropy encoding} of the (discrete) quantized symbols 
	requires an excess rate with respect to the rate-distortion bound not larger than 1.508 bits per complex-valued coefficient. 
	Hence, the ECSQ quantizer can achieve a rate of 
	\begin{equation}\label{eq:ECSQ_RD}
		R_{\text{ECSQ}} (D)= R_{\hv}^r(D) + 1.508 \, \widetilde{r}.  
	\end{equation}
	Note that when the target distortion is low (e.g., in the high-SNR regime where a large channel capacity is available), the overhead of $1.508 \, \widetilde{r}$ bits in quantizing an $MN$-dimensional channel becomes very small.  The following proposition states that the feedback scheme based on ECSQ yields the same QSE in estimating the channel as that of the optimal rate-distortion feedback. 

	\begin{proposition}\label{prop:ecsq_alpha}
		The ECSQ feedback strategy achieves the same QSE as the rate-distortion feedback strategy, i.e. $\alpha_{\rm ecsq} = \alpha_{\rm rd} $, where $\alpha_{\rm rd}$ is given in \eqref{eq:decay_exponent_ub}.
	\end{proposition}
	\begin{proof}
		See Remark \ref{remark:alpha_ECSQ} in Appendix \ref{app:rate_asymp_thm_proof}.
	\end{proof}
	The QSE of ECSQ is illustrated in Fig. \ref{fig:quality_scaling_exponent} (left), exactly identical to that of the rate-distortion feedback. 
	Accordingly, ECSQ suffices to achieve the (optimal) system DoFs of $1 + \alpha_{\rm rd} (K - 1)$ of the rate-distortion feedback. 

	\subsection{Analog Feedback}\label{sec:analog_feedback}
	
Feedback based on ECSQ requires full knowledge of the DL channel statistics at the UE in order to compute 
the MMSE estimate in \eqref{MMSE-u} and its KL expansion in \eqref{eq:KL_exp}. While the DL channel covariance 
can be estimated at the BS side from the UL pilots sent by the UE to the BS followed by a suitable UL-DL covariance transformation (e.g., see \cite{miretti2018fdd,khalilsarai2018fdd}), 
providing the DL channel statistics  to the UEs is significantly harder. 
The reason is that, as explained before, in massive MIMO  
we aim at minimizing the DL training dimension $\beta_{\rm tr}$ such that it is significantly less than $M$. 
Methods for covariance estimation from low dimensional projections (i.e., $\beta^{\rm tr} < M$) have been studied (e.g., see \cite{haghighatshoar2016massive}).  Nevertheless, they are computationally quite intensive. 
Therefore, it is preferable not to rely on the availability of DL channel covariance at the user side. 

The analog feedback (AF) strategy presented in this section is an example of a method that does not rely on such assumption.	
In AF, the UE extracts its $\beta_{\rm tr}$ received DL pilot symbols  from the DL training and feeds them back to the BS via Quadrature Amplitude Modulation (QAM) symbols with unquantized I and Q components via $\beta_{\rm fb} = \zeta \beta_{\rm tr}$ channel symbols (e.g., see \cite{thomas2005obtaining,marzetta2006fast}). 
In particular, the received training vector $\yv^{\rm tr}$ is modulated by a full-rank ``spreading'' matrix $\Psim$ of dimension $\beta_{\rm tr} \times \beta_{\rm fb}$.  The received feedback at the BS is given by 
	\begin{equation}\label{eq:af_bs_signal}
		\yv^{\rm af} = \yv^{\rm tr} \Psim + \zv^{\rm ul} = \hv^\herm \Xm^{\rm tr} \Psim +\zv^{\rm af},
	\end{equation}
	where $\zv^{\rm ul} \sim \cg (\mathbf{0},\mathbf{I})$ is the AWGN over the UL channel and $\zv^{\rm af} = \zv^{\rm tr}\Psim +\zv^{\rm ul}$. The scalar $\zeta=\beta_{\rm fb}/\beta_{\rm tr}$ denotes the number of feedback channel uses per training coefficient. 

The AF transmitted signal in \eqref{eq:af_bs_signal} consists of $\beta_{\rm fb}$ channel uses with symbols
$x_i^{\rm fb} = \yv^{\rm tr}\psiv_i$, where $\psiv_i$ is the $i$-th column of $\Psim$, for $i = 1, \ldots, \beta_{\rm fb}$.
From the UL per-user capacity scaling given in Section \ref{sec:ch_training},  the feedback transmit power per channel use is normalized such that
	\begin{equation}\label{eq:psi_vecs}
		\begin{aligned}
			\bE[|\yv^{\rm tr}\psiv_i|^2] = \psiv_i^\herm \Rm_{\yv^{\rm tr}} \psiv_i  = M \snrul,
		\end{aligned}
	\end{equation}
	for all $ i  =1,\ldots,\beta_{\rm fb}$, where 
	\[ \Rm_{\yv^{\rm tr}} = \bE [\yv^{\rm tr}\yv^{{\rm tr}\, \herm} ] = \Xm^{{\rm tr}\, \herm} (\Sigmam^h +\muv \muv^\herm)\Xm^{\rm tr} + \mathbf{I}\] 
	is the autocorrelation matrix of the receiver DL training signal. Note that selecting a set of $\beta_{\rm fb}$ vectors that satisfy \eqref{eq:psi_vecs} and which contain a subset of $\min(\beta_{\rm tr},\beta_{\rm fb})$ linearly independent elements is always possible because $\Rm_{\yv^{\rm tr}}$ 
	is of rank $\beta_{\rm tr}$. After receiving $\yv^{\rm af}$ given by (\ref{eq:af_bs_signal}), the BS computes the MMSE estimate of the channel given the feedback as 
	\begin{equation}
		\wh{\hv} = \bE\left[\hv | \yv^{\rm af} \right] = \Sigmam^h \Xm^{\rm tr} \Psim \Sigmam_{\yv^{\rm af}}^{-1} (\yv^{\rm af}-\muv^\herm \Xm^{\rm tr} \Psim)^\herm + \muv, 
	\end{equation}
	where $\Sigmam_{\yv^{\rm af}} = \Psim^\herm \Xm^{{\rm tr}\herm} \Sigmam^h \Xm^{\rm tr} \Psim +\Psim^\herm \Psim + \mathbf{I}  $. Note that, 
	unlike the rate-distortion quantizer and the ECSQ, AF does not need channel covariance knowledge at the UE and the processing at the UE is very simple, requiring only a matrix-vector multiplication. 
The CSIT estimation error with AF can be computed as
	\begin{equation}
		\begin{aligned}
			D &= \bE [ \Vert\hv - \widehat{\hv} \Vert^2 ]
			= {\rm Tr}\left( \Sigmam^h -  \Sigmam^h \Xm^{\rm tr} \Psim \Sigmam_{\yv^{\rm af}}^{-1} \Psim^\herm \Xm^{{\rm tr}\, \herm} \Sigmam^h\right).\\
		\end{aligned}
	\end{equation}
	The following theorem yields the scaling law of this error for large SNR.
	\begin{theorem}\label{thm:af_distortion_modified}
		The AF strategy achieves a channel estimation error of $\Theta (\snrdl^{-\alpha_{\rm af}})$ with probability one over the realizations of the training matrix $\Xm^{\rm tr}$, where the QSE $\alpha_{\rm af}$ is given by
		\begin{equation}\label{eq:af_decay_exponent}
			\alpha_{\rm af} = \mathbf{1}\{\min(\beta_{\rm tr},\beta_{\rm fb}) \geq r\}.
		\end{equation}
	\end{theorem}
	
	\begin{proof}
		See Appendix \ref{app:af_thm_proof_modified}.
	\end{proof}
	
	Fig. \ref{fig:quality_scaling_exponent} summarizes the QSE results seen so far. On the left, we show as a heat map the QSE achieved by the ECSQ, which by 
	Proposition \ref{prop:ecsq_alpha} coincides with the optimal QSE achieved by the rate-distortion scheme.  The QSE of AF is illustrated as a heat map 
	in Fig. \ref{fig:quality_scaling_exponent} (right). Comparing the right and left heat maps we notice that the QSE achieved by AF is optimal 
	for all training and feedback dimensions $(\beta_{\rm tr},\beta_{\rm fb})$ belonging to regions $\Rc_1$ and $\Rc_3$. In region $\Rc_2$, 
	rate-distortion feedback and ECSQ achieve an exponent of $\beta_{\rm fb}/r > 0$, whereas AF has exponent zero, and is therefore strictly sub-optimal.

	\begin{figure*}[t]
		\centerline{\includegraphics[scale=0.6]{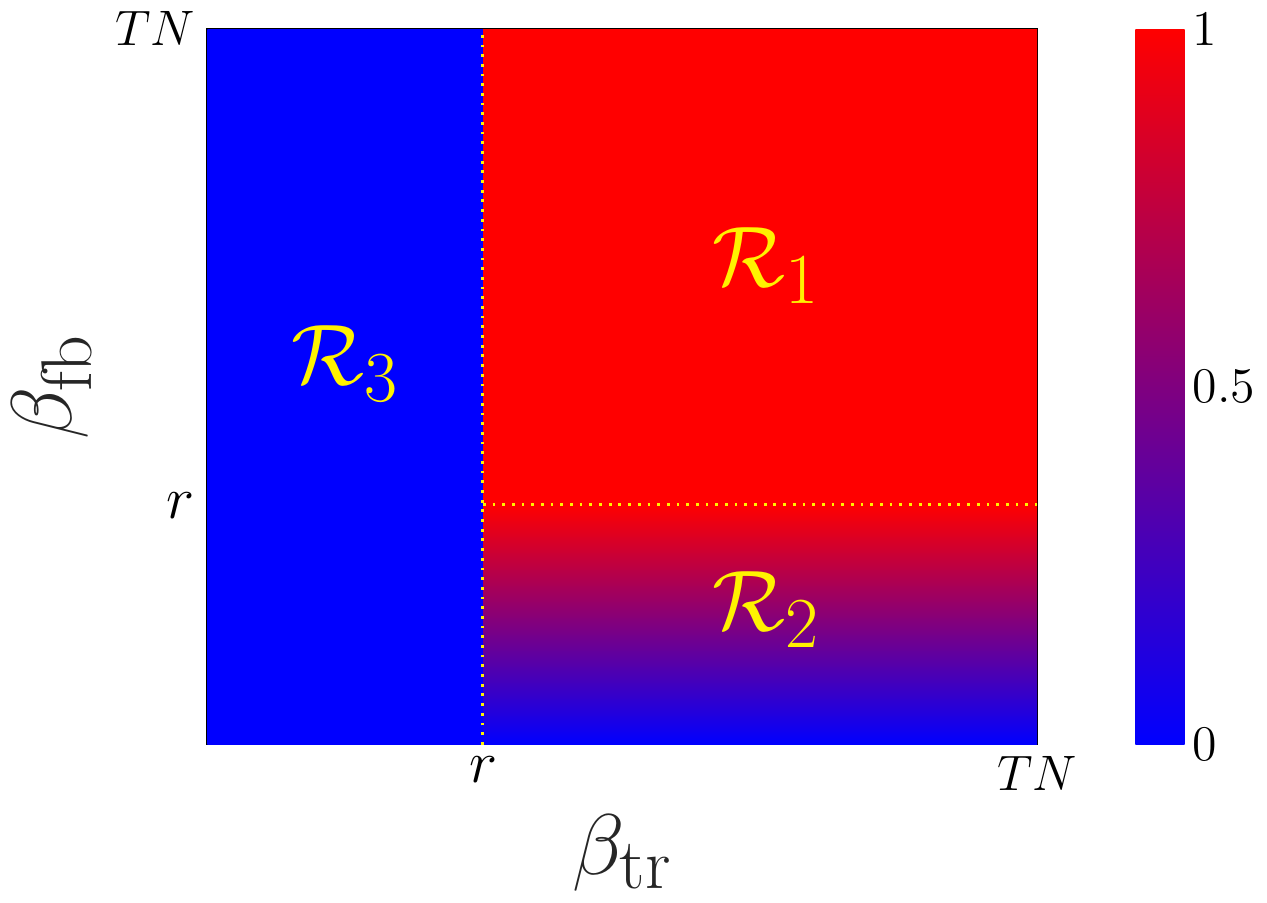}
			\includegraphics[scale=0.6]{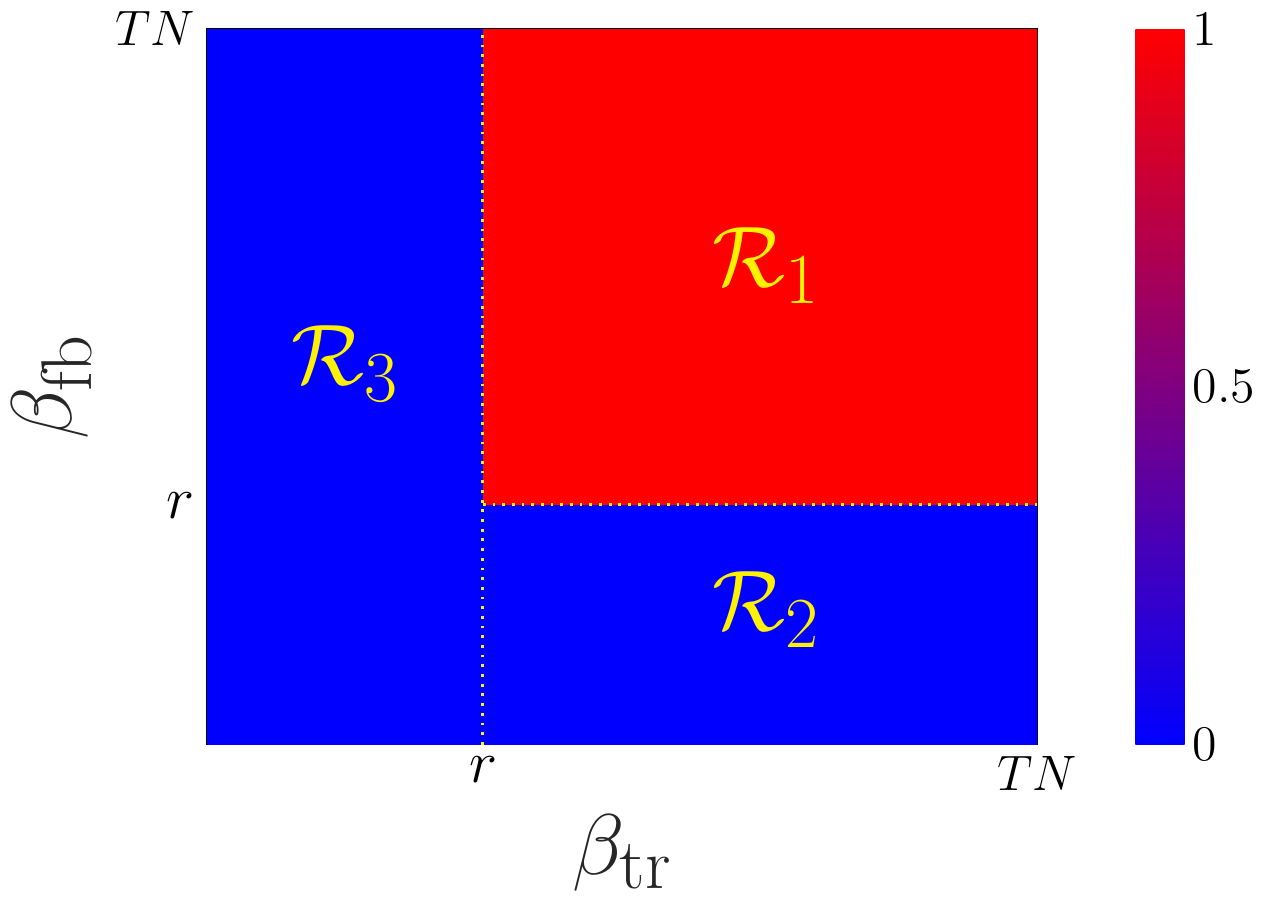}}
		\caption{The QSE as a function of $(\beta_{\rm tr},\beta_{\rm fb})$, represented as a heat-map. The left figure corresponds to the optimal rate-distortion feedback $\alpha_{\rm rd}$, which also coincides with the QSE $\alpha_{\rm ecsq}$ achieved by ECSQ, whereas the right figure corresponds to AF $\alpha_{\rm af}$.}
		\label{fig:quality_scaling_exponent}
	\end{figure*}

	\subsection{3GPP-Inspired Feedback Scheme}\label{sec:CS_based}
Release 16 of the the 3GPP new radio suggests a new wideband CSI feedback method that reduces overhead by sending only the strongest coefficients of the channel decomposition in terms of space-frequency basis vectors (see, e.g. \cite{3GPPsource1,3GPPsource2,ahmed2019overhead}). We describe the concept of this highly popular scheme in the industry, and based on it, introduce the {\em 3GPP-inspired} feedback method studied in this work.
Let $\Hm = [\hv[1], \ldots, \hv[N]]$ denote the $M \times N$ (antenna-subcarriers) channel matrix for a given UE, such that $\hv = {\rm vec}(\Hm)$. The scheme consists of approximating
	\begin{equation}\label{eq:H_tilde_expression}
		\Hm \approx \Am \Wm \Bm^\transp = \sum_{j=1}^{G_a} \sum_{i=1}^{G_d} w_{i,j} \av (\thetav_j) \bv (\tau_i)^\transp,
	\end{equation}
for fixed codebooks $\Am=[\av(\thetav_1),\ldots,\av(\thetav_{G_a})]$ and $\Bm = [\bv(\tau_1),\ldots,\bv(\tau_{G_d})]$, where $G_a$ is proportional to the number of antennas (typically $G_a =2M$) and $G_d$ is proportional to the number of subcarriers (typically $G_d=2N$). The columns of $\Am$ are the array steering vectors for a discrete grid of angle-of-departures (AoDs)  and the columns $\bv(\tau_i)$ of $\Bm$ form a Fourier basis for a discrete grid of delays.
The main idea of the scheme is that the channel is ``sparse'' in the angle-delay domain. Therefore, the matrix $\Wm$ contains only a few dominant coefficients. Hence, the scheme sends the quantization of such dominant coefficients and the indices of the corresponding codebook vectors as a representation of the channel. Reformulating \eqref{eq:H_tilde_expression} we have
\begin{equation}\label{eq:h_express}
	\begin{aligned}
		\hv & \approx \vec{(\Am \Wm \Bm^\transp)} = (\Bm \otimes \Am) \wv,
	\end{aligned}
\end{equation} 
where $\wv = \vec{(\Wm)}$ is the sparse vector containing the angle-delay channel coefficients and $\otimes$ is the Kronecker product. From \eqref{eq:training_signal} and \eqref{eq:h_express} we can write the training measurements as
\begin{equation}
	\yv^{\rm tr} \approx \wv^\herm (\Bm \otimes \Am)^\herm  \Xm^{\rm tr} + \zv^{\rm tr} = \wv^\herm \Phim^\herm +\zv^{\rm tr} ,
\end{equation}	
where we have defined $\Phim^\herm \triangleq (\Bm \otimes \Am)^\herm  \Xm^{\rm tr}$ as the \textit{sensing matrix}. The idea is to quantize and feedback the $s_0$ largest coefficients in $\wv$. But in order to send the dominant coefficients, first an estimate of $\wv$ must be obtained at the UE from the DL training 
$\yv^{\rm tr}$ in \eqref{eq:training_signal}. Interestingly, the 3GPP standard does not specify how such estimate should be obtained. 
Typical Least-Squares estimation in the form (e.g., see \cite{marzetta2010noncooperative})
\begin{equation} 
	\widehat{\wv} =  (\Phim^\herm \Phim)^{-1} \Phim^\herm ( \yv^{\rm tr} )^\herm  \label{LS}
\end{equation}
requires that the $MN \times MN$ matrix $\Phim^\herm \Phim$ is invertible, which in turns requires
$\beta_{\rm tr} \geq MN$, which is completely impractical in massive MIMO. Heuristic approaches consist of limiting the estimation only on the pilot subcarriers in $\Nc_p$ (typically one per resource block) and then interpolating over the subcarriers (e.g., see \cite{shirani2009channel}). In this case, the invertibility conditions yields $\beta_{\rm tr} \geq M N_p$, which yields
$T_p \geq M$ pilot symbols in time per pilot subcarrier. For a typical resource block of about 200 symbols and number of antennas ranging from 32 to 128 it is clear that the DL pilot overhead is still too large. For this reason, several works have focused on 
exploiting the inherent sparsity of $\wv$ to estimate it at the UE when $\beta_{\rm tr}$ is smaller than $MN$. These approaches use 
some form of CS, and are referred to in general as ``compressed DL pilot'' schemes (e.g., see \cite{bajwa2008compressed,zhang2018distributed,shen2016compressed}). Now, there exist plenty of sparse recovery methods in the CS literature to estimate $\wv$ from the noisy measurements $\yv^{\rm tr}$. Among these, we consider 
the well-known orthogonal matching pursuit (OMP) \cite{tropp2007signal} as a representative of sparse recovery methods, first because it is shown to be highly successful in a variety of settings and second because it is much less computationally complex in comparison to the convex optimization-based alternatives such as \cite{chen2001atomic}. We omit the details of OMP to save space and refer the reader to \cite{tropp2007signal}.

To perform OMP, the user is assumed to know the sensing matrix $\Phim$ (composed of the training matrix and the codebooks) and the sparsity level $s_0$. The estimator outputs an estimate $\widehat{\wv}$ containing $s_0$ non-zero coefficients. Following \cite{3GPPsource1}, scalar quantization is performed independently on the normalized 
amplitude and phase of each coefficient. The amplitudes are first normalized to the largest one and then each quantized with a codebook of $2^b$ uniformly spaced quantization levels in the interval $[0,1]$ as $\Qc_{\rm amp} = \{ 2^{-b}\ell\, :\, \ell=1,\ldots, 2^b \}$, for some positive integer $b$. The phase is quantized with a codebook of $2^b$ uniformly spaced quantization levels in the interval $[0,2\pi]$ as $\Qc_{\rm phase} = \{ 2\pi \times 2^{-b}\ell\, :\, \ell=1,\ldots, 2^b \}$. Therefore $2b$ bits are spent to quantize a single coefficient. Finally, the quantization bits and the indices of the codebook columns whose corresponding coefficients are non-zero (known as the support set) are transmitted to the BS.

The resulting  total number of feedback bits in this case is given by
\begin{equation}\label{eq:cs_bits}
	B_{\rm cs} = 2b (s_0-1) + \lceil \log s_0 \rceil + s_0 \lceil \log (G_a G_d)\rceil.  
\end{equation}
Here the first, second and third terms represents the number of bits spent on quantizing amplitudes, the number of bits spent to report the number of non-zeros and the number of bits necessary to encode the support set, respectively. From this, we can compute the feedback dimension necessary with the OMP method and scalar 
quantization as $\beta_{\rm fb} = B_{\rm cs}/C^{\rm ul}$.

The channel estimation error at the BS for CS-based feedback does \textit{not} go to zero even in high SNR, because of the approximation error inherent in \eqref{eq:h_express}, which is independent of SNR. The point is that, except in rare cases, the channel $\hv$ cannot be exactly described by a linear combination of $s_0$ columns of $\Bm\otimes \Am$, since the multipath angle-delay parameters do not lie on the assumed discrete grid but are rather arbitrarily distributed over the continuum. This mismatch between the exact ``off-grid" sparse representation of a vector and its on-grid sparse approximation is well-known in the CS literature (see, e.g. \cite{chi2011sensitivity}). The mismatch results in an error in sparse estimation of the channel, which does not vanish with SNR. Therefore, not only OMP, but any CS-based feedback method based on approximation of the channel with $s_0$ mismatched angle-delay vectors (with $s_0 \le MN$) results in a QSE of $\alpha_{\rm cs}=0$ for all $\beta_{\rm tr}< MN$ and all $\beta_{\rm fb}$.\footnote{If $\beta_{\rm tr}\ge MN$, the UE can simply 
compute the  Least-Squares estimate \eqref{LS} with an error that goes to zero with $\snrdl\to \infty$.} 
However, as we shall see in the next section, the {\em 3GPP-inspired} scheme can achieve a DL spectral efficiency in line with the other schemes when the channel is sufficiently sparse.

\section{Numerical Results}\label{sec:numerical_results}
In this section we present simulation results that confirm our theoretical findings and compare the studied feedback strategies for two different channel models. 
\begin{enumerate}
	\item \textbf{Synthetic Multipath Channel Model:} The synthetic multipath wideband channel model for an arbitrary user is given by 
	\[ \hv = \sum_{\ell=1}^L c_\ell \bv (\tau_\ell)\otimes \av (\theta_{\ell}), \]
	where $\av (\theta_{\ell})$ is the ULA steering vector at angle $\theta_{\ell}$ and $\bv (\tau_\ell)$ contains the phase rotations corresponding to the path delay $\tau_\ell$ with elements $[\bv (\tau_{\ell})]_n = e^{-j2\pi n f_s \tau_{\ell}}$. This expression is similar to \eqref{eq:h_express}, except that we assume each angle $\theta_{\ell}$ to be generated uniformly at random over $[-\pi/2,\pi/d]$ and each delay $\tau_\ell$ to be generated at random over $[0,\tau_{\max}]$. In addition, each path gain $\{c_\ell\}_{\ell}$ is distributed as $\cg (0,1)$. The number of paths $L$ is the same for all UEs. This results in a correlated Gaussian channel with mean $\muv = 0$ (which is intended by construction for simplicity) and covariance $\Sigmam^h=\sum_{\ell=1}^L (\bv (\tau_\ell)\bv (\tau_\ell)^\herm )\otimes (\av (\tau_\ell) \av (\tau_\ell)^\herm)$. This multipath channel model is considered extensively in the literature \cite{sayeed2003virtual,bajwa2008compressed,wunder2018hierarchical,haghighatshoar2017massive} and ensures a correlated Gaussian distribution for the channel, 
	which is a prerequisite for our theoretical results to hold.
	\item \textbf{CDL Channel Model:} The clustered delay line (CDL) channel types are proposed by 3GPP as standard link-level MIMO fading channel models \cite{gpptr38901,3gpp2020study}. In these models the channel is  the contribution of a number of multipath ``clusters", parametrized by their average gain power, AoDs and delays.
	In order to generate CDL model, we use the MATLAB system object \texttt{nrCDLChannel} \cite{MATLAB:2020}. For each user, we generate $r$ clusters with angle and delay parameters chosen uniformly at random and independently across users.
	The system object produces random realizations of the fading channel, which we normalize such that $\Vert \hv \Vert^2=MN$. We compute an approximation of the channel covariance by computing the sample covariance matrix over 2000 realizations.\footnote{Note that this number of samples is sufficient for the sample covariance matrix to reasonably converge to the true covariance, given the small covariance rank in these simulations.} Note that in this model the channel does not follow a Gaussian distribution, except perhaps in an approximate sense due to the summation of many random multipath gains as a consequence of the Central Limit Theorem. However, the simulations will show that the ranking of the discussed feedback strategies and the theoretical results of this work not only hold for correlated Gaussian channels but also empirically carry over to the more realistic CDL model.
\end{enumerate}

\subsection{Physical Channel and OFDM Parameters}

In all simulations we consider a uniform linear array (ULA) with $M=32$ antennas at the BS communicating with $K=6$ UEs over a total of $N=32$ OFDM subcarriers. This results in a wideband channel of dimension $MN=1024$. We consider an OFDM subcarrier spacing of $f_s = 30$ kHz, with the useful symbol duration equal to $T_u = 1/f_s\approx 33$ $\mu$s. We assume the maximum channel delay spread to be equal to $\tau_{\max} = 7$ $\mu$s and we take the duration of the OFDM cyclic prefix  to be equal to the maximum delay spread for simplicity, which results to an OFDM symbol duration of $T_s = T_u + \tau_{\max} = 40 $ $\mu$s. 
The corresponding coherence bandwidth is given by $B_c = 1/\tau_{\max} \approx 143 $ kHz, over which the channel has relatively small variation, while the total signal bandwidth is equal to $BW = Nf_s\approx 1$ MHz. The channel is assumed to be constant over a coherence time of $T_c = 1$ ms, corresponding to $T = T_c / T_s =25$ OFDM symbols in time. 
Therefore the wideband channel $\hv$ is an $MN$-dimensional vector process that is constant over a temporal frame of $T$ OFDM symbols and is independently generated over distinct frames according to the same statistics. The total time-frequency resource dimension of $\beta = TN = 25\times 32 = 800$ consists of $\beta_{\rm tr}$ dimensions dedicated to channel training and $\beta - \beta_{\rm tr}$ dimensions for data transmission. The training dimensions are given by considering $N_p$ pilot subcarriers and sending a pilot sequence of length $T_p$ generated according to \eqref{eq:isotropic_pilots} over each of them such that $N_p T_p = \beta_{\rm tr}$. In all simulations, we consider the SNR in UL to be 10 dBs less than the SNR in DL, i.e. $\snrul = \kappa \snrdl$ with $\kappa = 0.1$. In the experiments, all feedback methods use the same number of feedback dimensions to enable a fair comparison.  

\subsection{Performance Metrics}

We study the CSI estimation MSE and the multiuser system total spectral efficiency (sum-rate) with ZF precoding for all the feedback strategies. To compute these, we first fix a channel distribution according to one of the models described above and generate a set of random, wideband channel realizations based on that distribution. Then we generate a single random training matrix according to the design explained in Section \ref{sec:ch_training}. We then produce the noisy training measurements, compute the feedback according to each of the schemes and estimate channel accordingly, for each random channel realization. As a metric of estimation error, we consider the average, normalized MSE defined as 
\begin{equation}\label{eq:normalized_MSE}
	\begin{aligned}
		\text{MSE}_{\text{avg}}  = \frac{1}{MN} \bE[\Vert \hv- \widehat{\hv}  \Vert^2],
	\end{aligned}
\end{equation}
where the expected value is computed empirically by averaging over 
the randomly generated training matrices, channel distributions and random channels realizations. To compare achievable sum-rates in DL, we consider ZF precoding where the transmit data vector over subcarrier $n$ is given by $\dv [n] = \sum_{k=1}^K \sqrt{P_k} s_k [n] \vv_k [n]$, where $P_k = \snrdl /K$ is the (uniform) transmission power per user, $s_k [n]\in \bC$ is the data symbol intended for UE $k$ such that $\bE[|s_k [n]|^2]\le 1$, and $\{\vv_k [n]\}_{k=1}^K$ are the precoding vectors, given by the column-normalized pseudo-inverse of the estimated channel matrix.
Defining the variables $g_{k,k'}[n] = \sqrt{P_{k'}} (\hv_{k}[n]^\herm \vv_{k'}[n])$, we can write the achievable ergodic rate for user $k$ at subcarrier $n$ as \cite{caire2018ergodic}
\begin{equation}
	R\, [k,n] =\bE \left[\log \left(1+ \frac{|g_{k,k}[n]|^2}{1 + \sum_{k'\neq k} |g_{k,k'}[n]|^2}\right) \right].
\end{equation}
We consider the average sum-rate over all subcarriers, which is defined as
\begin{equation}\label{eq:ach_rate}
	R_{\text{avg}} = \frac{T}{\beta} \sum_{n\in \Nc_p^c,k} R[k,n] + \frac{T-T_p}{\beta} \sum_{n\in \Nc_p,k} R[k,n] ~~ [\text{bits}/\text{s}/\text{Hz}],
\end{equation}
where $\Nc_p^c$ is the complement of the pilot subcarrier index set. Note that this weighted averaging is necessary to take into account the fact that over the pilot subcarriers, data is transmitted over $T-T_p$ out of $T$ symbols, whereas on other subcarriers all $T$ symbols are used to send data. 

\begin{figure*}
	\centering
	\begin{subfigure}{.45\textwidth}
		\centering
		\includegraphics[ scale=0.6]{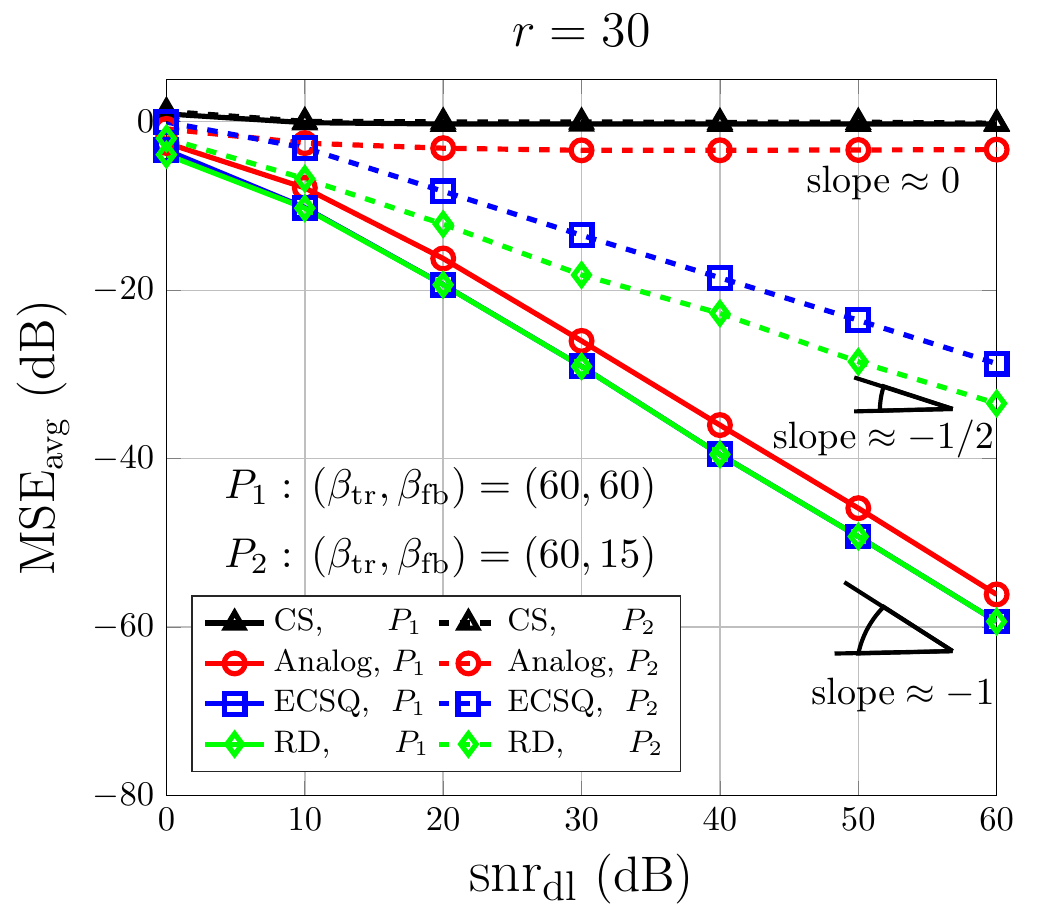}
		\caption{synthetic multipath channel model}
		\label{fig:MSE_vs_SNR_r_30}
	\end{subfigure}%
	\begin{subfigure}{.45\textwidth}
		\centering
		\includegraphics[scale=0.6]{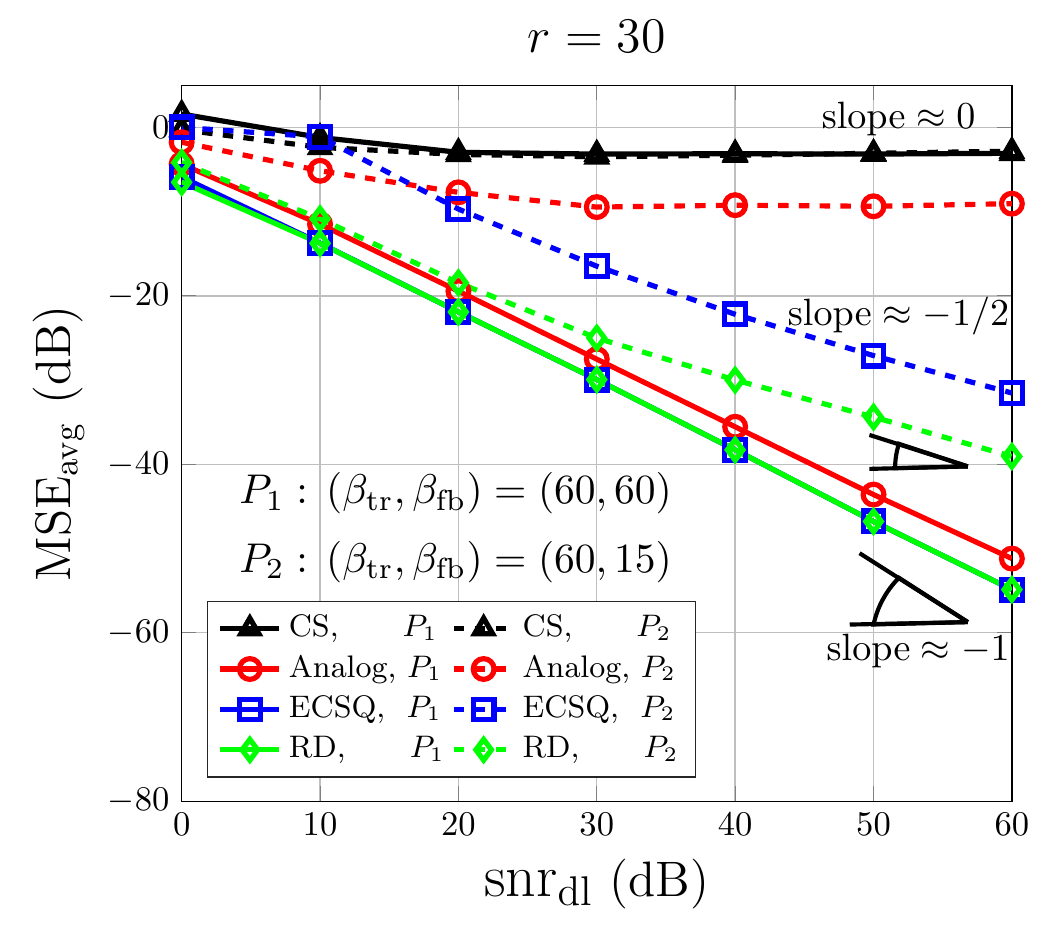}
		\caption{CDL channel model}
		\label{fig:MSE_vs_SNR_r_30_CDL}
	\end{subfigure}
	\caption{\small Comparison of channel estimation average normalized MSE vs. SNR for synthetic multipath and CDL channel models, where $r=30$.}
	\label{fig:MSE_vs_SNR}
\end{figure*}

\subsection{Experiments}
In the first set of experiments, we plot the the average normalized MSE (in dBs) vs. the DL SNR for a channel covariance rank for synthetic multipath and CDL channel models in Figs. \ref{fig:MSE_vs_SNR_r_30} and \ref{fig:MSE_vs_SNR_r_30_CDL}, respectively. Here the channel covariance rank (equivalent to the number of paths) is set to $r=30$ for all users. We compare four feedback strategies: the rate-distortion (RD) feedback, which is theoretically optimal, ECSQ feedback, analog feedback and the 3GPP-inspired feedback scheme (denoted here by CS). We consider two different points in the $\beta_{\rm tr}-\beta_{\rm fb}$ plane to illustrate the effect of different training and feedback dimensions on the CSIT estimation error. The solid lines in both figures correspond to the training and feedback dimension pair $(\beta_{\rm tr},\beta_{\rm fb}) = (60,60)$ which belongs to the region $\Rc_1$ (see Fig. \ref{fig:quality_scaling_exponent}). The dashed lines on the other hand correspond to $(\beta_{\rm tr},\beta_{\rm fb}) = (60,15)$ which belongs to the region $\Rc_2$. When $(\beta_{\rm tr},\beta_{\rm fb}) = (60,60)$, from \eqref{eq:decay_exponent_ub} and \eqref{eq:af_decay_exponent} we expect that when RD, 
the ECSQ and analog feedback methods achieve a QSE of $\alpha = 1$, since $\beta_{\rm tr} = \beta_{\rm fb} > r$.  In contrast, when $(\beta_{\rm tr},\beta_{\rm fb}) = (60,15)$,  RD and ECSQ achieve a QSE of $\beta_{\rm fb}/r = 15/30 = 1/2$, while analog feedback yields a QSE of 0. These predictions are confirmed by the average MSE vs. SNR curves of Figs. of both \ref{fig:MSE_vs_SNR_r_30} and \ref{fig:MSE_vs_SNR_r_30_CDL}, where the slope of the curves in high SNR is equivalent to $-\alpha$. Notice that in both figures, the 3GPP-inspired scheme yields a high MSE and one that remains constant in high SNR, confirming 
the fact that $\alpha_{\rm cs} = 0$. The error decay exponent is zero, because of the explained mismatch between the true space-delay sparsity domain and the discretized domain assumed by the estimator. Besides, the error is high due to the fact that the user receives $\beta_{\rm tr}=60$ pilots, and estimates a channel with sparsity level $r=30$. For this sparsity level, the CS sparse recovery method needs far more measurements to achieve a low estimation error. 
\begin{figure*}
	\centering
	\begin{subfigure}{.45\textwidth}
		\centering
		\includegraphics[ scale=0.6]{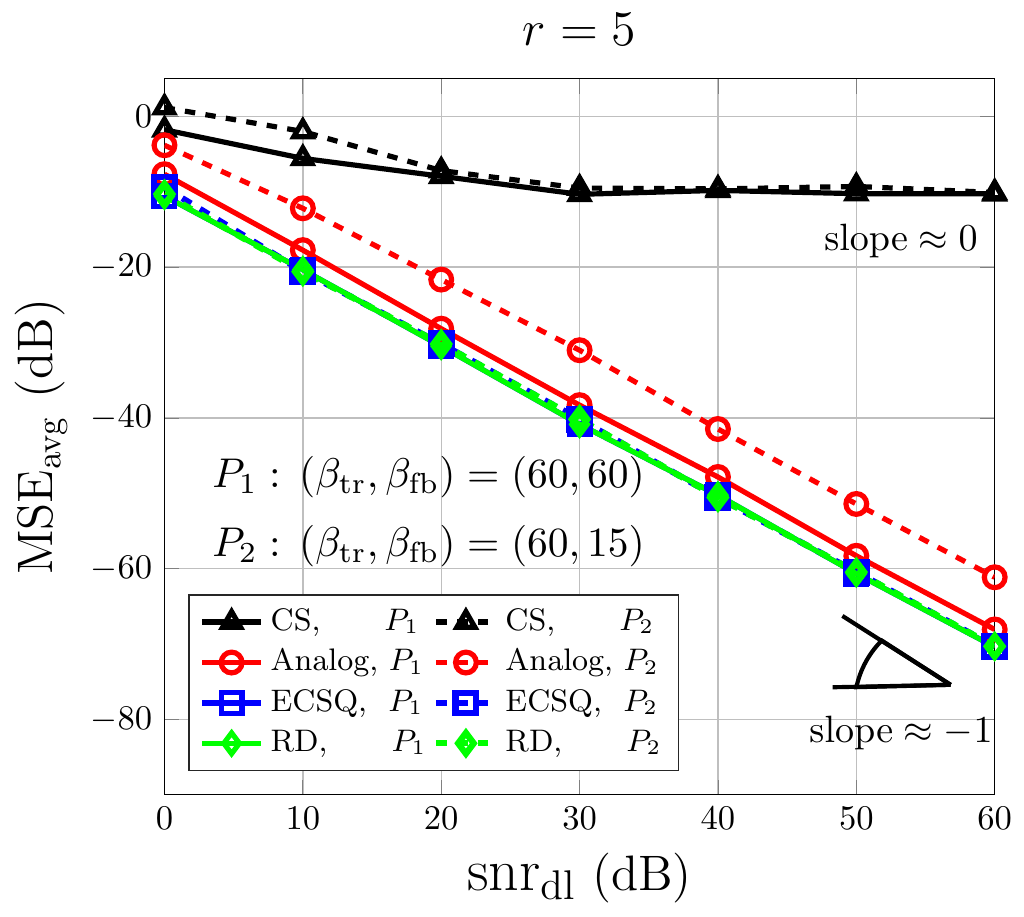}
		\caption{synthetic multipath channel model}
		\label{fig:MSE_vs_5_r_30}
	\end{subfigure}%
	\begin{subfigure}{.45\textwidth}
		\centering
		\includegraphics[scale=0.6]{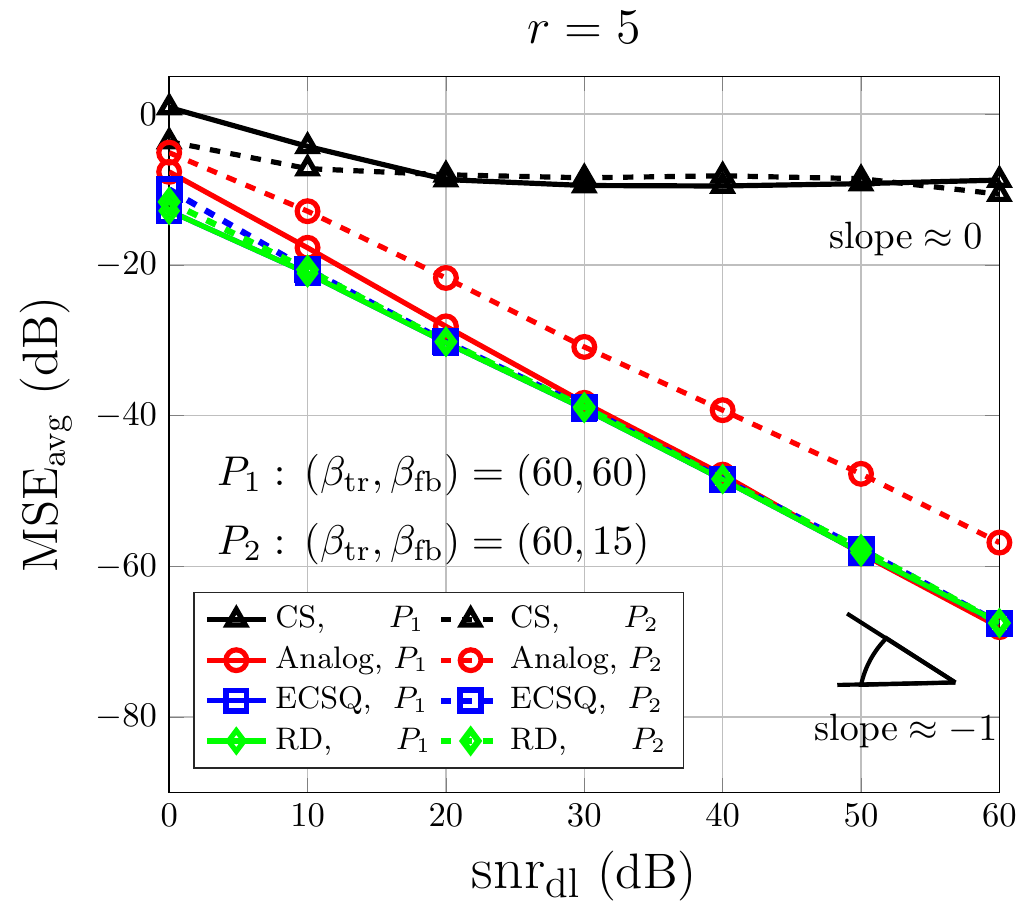}
		\caption{CDL channel model}
		\label{fig:MSE_vs_SNR_r_5_CDL}
	\end{subfigure}
	\caption{\small Comparison of channel estimation average normalized MSE vs. SNR for synthetic multipath and CDL channel models, where $r=5$.}
	\label{fig:MSE_vs_SNR_r_5}
\end{figure*}

We highlight the latter point by repeating the experiment for a sparser channel with $r=5$ multipath components. The results are given in Figs. \ref{fig:MSE_vs_SNR_r_5} and \ref{fig:MSE_vs_SNR_r_5_CDL}, where we observe that the estimation error of the 3GPP-inspired feedback has significantly decreased in both channel models reaching a normalized MSE of around -10 dB. Since in this case the channel is much more sparse, the CS estimator at the user can estimate it with much less error. Furthermore, in both figures the estimation error with RD, ECSQ and analog feedbacks decrease with a slope of -1 in large SNR, equivalent to a QSE of 1. The reason is that when $r=5$, for both points $(\beta_{\rm tr},\beta_{\rm fb})\in \{(60,60),(60,15)\}$, we have $\beta_{\rm tr}>r$ and $\beta_{\rm fb}>r$ in which case the above mentioned methods achieve a QSE of 1 as predicted by  the theoretical results.

\begin{figure*}
	\centering
	\begin{subfigure}{.45\textwidth}
		\centering
		\includegraphics[scale=0.6]{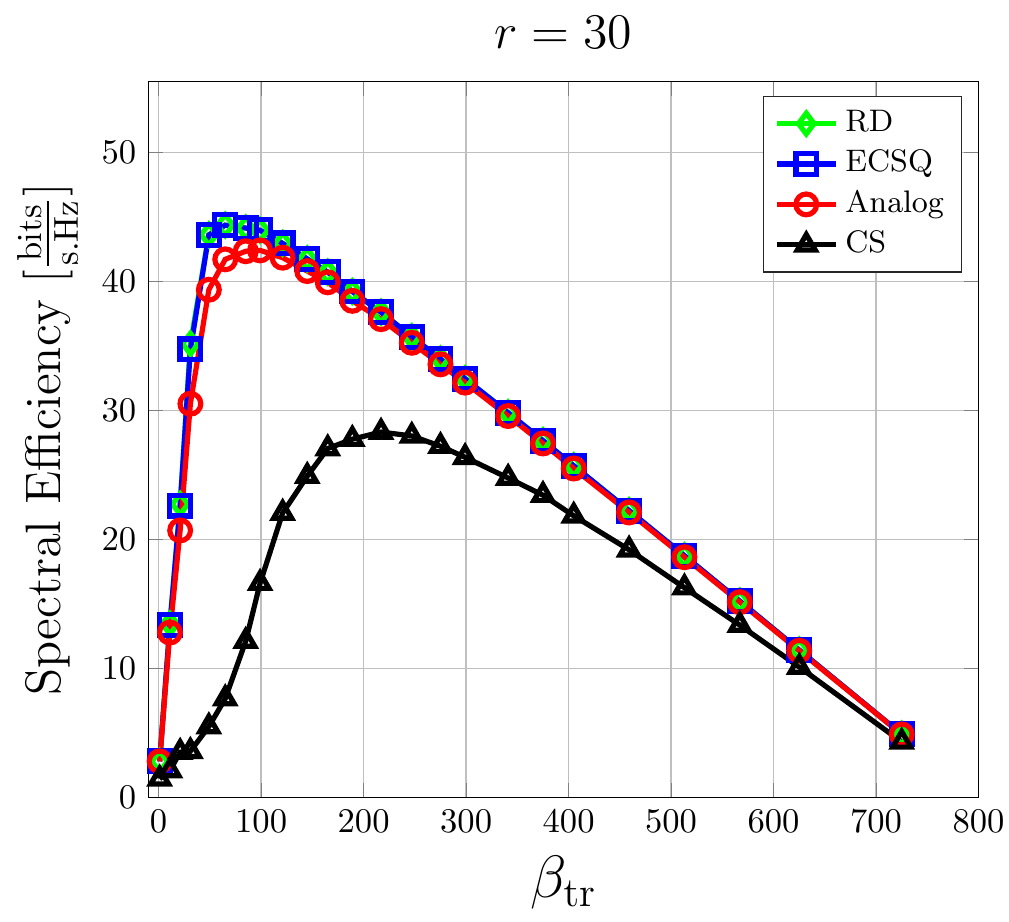}
		\caption{synthetic multipath channel model}
		\label{fig:Rate_vs_Beta_r_30}
	\end{subfigure}%
	\begin{subfigure}{.45\textwidth}
		\centering
	\includegraphics[scale=0.6]{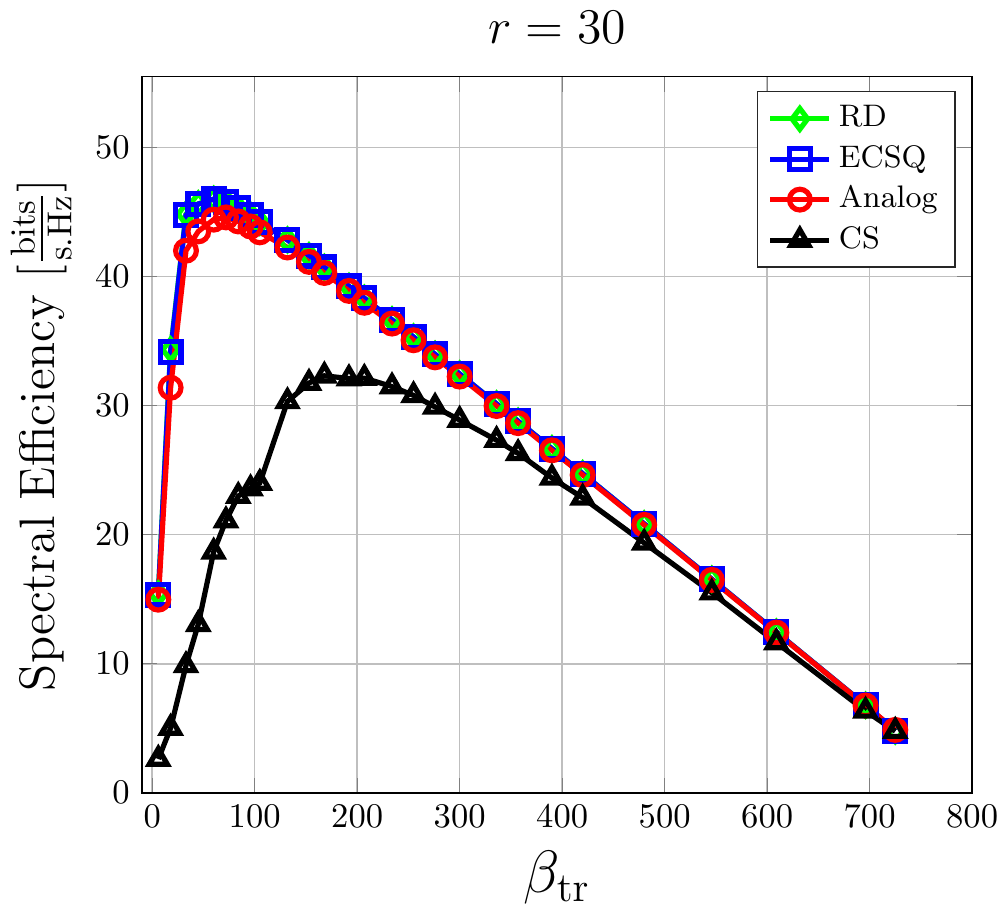}
	\caption{CDL channel model}
	\label{fig:Rate_vs_Beta_r_30_CDL}
	\end{subfigure}
%
	\caption{\small Comparison of channel estimation average normalized MSE vs SNR for synthetic multipath and CDL channel models. Here we have $r=30$ and the SNR in DL is set to $\snrdl = 20$ dB.}.
	\label{fig:Rate_vs_Beta}
\end{figure*}

In the next experiment, we compare the DL sum-rate in \eqref{eq:ach_rate} achieved by the four feedback methods. Figs. \ref{fig:Rate_vs_Beta_r_30} and \ref{fig:Rate_vs_Beta_r_30_CDL} illustrate the results for the synthetic multipath and CDL channel models. Here we have considered training and feedback dimensions to be equal at each point of the curves, so that the horizontal axis represents both training and feedback dimensions. The downlink SNR is set to $\snrdl = 20$ dB. Note that for both very small $(\beta_{\rm tr}\to 0)$ and very large $(\beta_{\rm tr}\to TN=800)$ training dimensions, the sum-rate is close to zero, since in the first case the channel estimation error is very high due to insufficient training, resulting in high interference and in the second case most signal dimensions are spent on training and very little on data transmission. From \eqref{eq:decay_exponent_ub}, \eqref{eq:af_decay_exponent} and Proposition \ref{prop:ecsq_alpha} we know that the RD, ECSQ and analog feedback schemes have the same QSE whenever $\beta_{\rm tr}=\beta_{\rm fb}$ and therefore the same DoF. From this, we expect that the sum-rate achieved by these three schemes to be very close to each other in high SNR. The curves of Figs. \ref{fig:Rate_vs_Beta_r_30} and \ref{fig:Rate_vs_Beta_r_30_CDL} show that even for a moderate SNR of $20$ dBs, the three methods achieve very similar downlink sum-rates. The implication is that the one-shot ECSQ and analog feedback schemes are sufficiently close to the optimal RD feedback in terms of the achievable downlink sum-rate. Furthermore, as we can see the 3GPP-inspired feedback method yield a significantly lower sum-rate. The reason is that, channel estimation via compressed sensing requires a training dimension that is relatively much larger than the other methods. Therefore it can be competitive, only if the channel is very sparse. Otherwise, the larger training dimension needed for an accurate channel estimation via compressed sensing results in a significant penalty in terms of the DL training-data transmission trade-off which is exemplified in Figs. \ref{fig:Rate_vs_Beta_r_30} and \ref{fig:Rate_vs_Beta_r_30_CDL}.

\section{Conclusion}
We provided optimal rate-distortion bounds for the problem of CSIT feedback in wideband massive MIMO systems, demonstrating an upper-bound on the MSE decay rate in high SNR for any feedback scheme when the BS broadcasts random Gaussian training pilots to the users. We then discussed three one-shot feedback methods, each assuming various levels of channel statistics knowledge either at the BS or the users. In particular, we showed that the low-complexity analog feedback yields a near-optimal high-SNR channel estimation error decay with no channel statistics knowledge at the user side and no quantization quantization and channel coding. We also studied a 3GPP-inspired feedback method based on compressed sensing estimation of the channel sparse coefficients at the user side, where we showed that the method entails a residual CSIT estimation error even in high SNR, due to a mismatch between the true space-delay domain of sparsity and the one assumed by the estimator. However, the method can yield decent results for sufficiently sparse channels. The findings were supported by numerical simulations comparing the normalized channel estimation MSE and achievable DL spectral efficiency with various feedback methods.

\appendices


\section{Proof of Lemma \ref{lem:rate_distortion}}\label{app:RD_lemma_proof}

We start by stating a few standard results regarding the (remote) rate-distortion function. It is well-known that the rate-distortion function of an i.i.d source represented by the random variable $\uv$ with distribution $p_{\uv}$ can be computed as (see \cite{thomas2006elements} Theorem 10.2.1)
\begin{equation}\label{eq:rate_distortion}
	R_{\uv} (D) = \min_{p_{\widebar{\uv}|\uv}: d(\widebar{\uv},\uv)\le D} I(\widebar{\uv};\uv),
\end{equation}
where $\widebar{\uv}$ is the quantization of $\uv$, $I(\widebar{\uv};\uv)$ is the mutual information between $\uv$ and $\widebar{\uv}$, $d(\widebar{\uv},\uv)$ is the distortion between $\uv$ and $\widebar{\uv}$ (see \eqref{eq:err_def}) and the minimum is taken over all conditional distributions for which the joint distribution $p_{\uv,\widebar{\uv}} $ satisfies the distortion constraint. It is also known that the remote rate distortion function of a source represented by the random variable $\hv$, and encoded given its observations denoted by the random variable $\yv^{\rm tr}$ is given by (see \cite{berger1971rate})
\begin{equation}\label{eq:remote_rate_distortion}
	R_{\hv}^r (D) = \min_{p_{\widebar{\hb}|\yb^{\rm tr}} : d(\widebar{\hv},\hv)\le D}~I(\widebar{\hv};\yv^{\rm tr})
\end{equation}
where $\widebar{\hv}$ is the quantization of $\hv$, $I(\widebar{\hv};\yv^{\rm tr})$ is the mutual information between $\widebar{\hv}$ and $\yv^{\rm tr}$, and the minimum is taken over all conditional distributions $p_{\widebar{\hv}|\yv^{\rm tr}}$ for which the joint distribution $p_{\hv,\widebar{\hv}}= p_{\hv} p_{\yv^{\rm tr}|\hv}p_{\widebar{\hv}|\yv^{\rm tr}} $ satisfies the distortion constraint. Note that since all sources are i.i.d, we have removed realization index superscripts from the variables (hence $\hv$ instead of $\hv^{(i)}$). From the premise of the lemma, $\uv$ is the MMSE estimate of the channel given the training measurements, i.e. $\uv = \bE [\hv|\yv^{\rm tr}]$. Using the same technique employed to prove inequality (15) of \cite{eswaran2019remote} (see Appendix A in \cite{eswaran2019remote}), we can verify that the remote rate-distortion function of $\hv$ is related to the rate distortion function of $\uv$ by
\begin{equation}\label{eq:relation_1}
	R_{\hv}^r (D)  = R_{\uv}(D-D_{\rm mmse}),
\end{equation}
for $D\ge D_{\rm mmse}$, where $D_{\rm mmse} = \bE \left[\Vert \hv - \uv\Vert^2 \right]$ is the MMSE of estimating the channel at the UE. 

On the other hand, the rate-distortion function of a correlated vector Gaussian source is given by reverse water-filling over its covariance eigenvalues \cite{cover2006elements}. If we denote the eigenvalues of $\Sigmam^u$ by $\{ \lambda_{\ell}^u \}_{\ell=1}^{MN}$, then we have $R_{\uv} (D) = \sum_{\ell=1}^{MN}\left[ \log \frac{\lambda_{\ell}^u}{\gamma} \right]_+,$ where $\gamma $ is chosen such that $\sum_{\ell=1}^{MN}\min \{ \gamma, \lambda_{\ell}^u \}=D$. Plugging this in \eqref{eq:relation_1} we get
\begin{equation}
	R_{\hv}^r (D) = \sum_{\ell=1}^{MN}\left[ \log \frac{\lambda_{ \ell}^u}{\gamma} \right]_+,
\end{equation}
where $\gamma$ is chosen such that $\sum_{\ell=1}^{MN}\min \{ \gamma, \lambda_{u, \ell} \} = D-D_{\rm mmse}.$ The proof is complete. \hfill $\blacksquare$

\section{Proof of Theorem \ref{thm:rate_asymp}}\label{app:rate_asymp_thm_proof}

We divide the proof to two parts. First, we show that if $\beta_{\rm tr} < r$, then the achievable error behaves as $\Theta (1)$ for all realizations of $\Xm^{\rm tr}$. Second, we show that if $\beta_{\rm tr} \ge r$, then an error decaying as $\Theta (\snrdl^{-\min(\beta_{\rm fb}/r,1)})$ is achievable with probability 1 
over the realizations of $\Xm^{\rm tr}$.\\

\noindent\textbf{Part I.} To prove part I, we first bound the minimum mean squared error (MMSE) of estimating the channel given the training measurements at the user, namely the variable $D_{\rm mmse}$. From $\uv = \bE [\hv | \yv^{\rm tr}]$ we have
\begin{equation}\label{eq:Dmmse}
	\begin{aligned}
		D_{\rm mmse} = \bE[\Vert \hv - \uv \Vert^2] &= \trace \left(\bE [\hv\hv^\herm ] - \bE[\hv \yv^{\rm tr} ]\bE[\yv^{{\rm tr}\, \herm} \yv^{\rm tr} ]^{-1} \bE[\hv \yv^{\rm tr} ]^\herm  \right)\\
		& = \trace \left( \Sigmam^h - \Sigmam^h \Xm^{\rm tr} \left(\Xm^{{\rm tr}\, \herm}\Sigmam^h \Xm^{\rm tr} + \mathbf{I} \right)^{-1}\Xm^{{\rm tr}\, \herm} \Sigmam^h   \right)
	\end{aligned}
\end{equation}
Let us define the eigendecomposition of $\Sigmam^h$ as $\Sigmam^h = \Um_h \Lambdam_h \Um_h^\herm$, where $\Um_h\in \bC^{MN\times r} $ is a tall unitary matrix and $\Lambdam_h = {\rm diag}(\lambdav) \in \bR^{r\times r}$ is a diagonal matrix of positive eigenvalues represented by the vector $\lambdav=[\lambda_1,\ldots,\lambda_r]^\transp$ where we assume $\lambda_{1}\ge \ldots\ge \lambda_{r}$ without loss of generality. Using this decomposition and applying the Sherman-Morrison-Woodbury matrix identity to $\left( \Xm^{{\rm tr}\, \herm }\Sigmam^h \Xm^{\rm tr} + \mathbf{I} \right)^{-1}$, we have
\begin{equation}
	\begin{aligned}
		\Sigmam^h \Xm^{\rm tr} \left(\Xm^{{\rm tr}\, \herm}\Sigmam^h \Xm^{\rm tr} + \mathbf{I} \right)^{-1}\Xm^{{\rm tr}\, \herm} \Sigmam^h=\Um_h \Lambdam_h^{1/2}\Gm \Lambdam_h^{1/2}\Um_h^\herm 
	 - \Um_h\Lambdam_h^{1/2}\Gm \left( \mathbf{I} + \Gm \right)^{-1}\Gm\Lambdam_h^{1/2}\Um_h^\herm,
	\end{aligned}
\end{equation}
where we have defined 
\begin{equation}\label{eq:G_def_0}
	\begin{aligned}
		\Gm \triangleq \Lambdam_h^{1/2}  \Um_h^\herm \Xm^{\rm tr} \Xm^{{\rm tr}\, \herm }\Um_h  \Lambdam_h^{1/2}. 
	\end{aligned}
\end{equation}
Plugging this into \eqref{eq:Dmmse} we have
\begin{equation}\label{eq:Dmmse_2}
	D_{\rm mmse} = \trace \left( \Lambdam_h \left(\mathbf{I} -  \Gm + \Gm ( \mathbf{I}+\Gm)^{-1} \Gm \right)\right). 
\end{equation}
Using a simple trace inequality, one can show that 
\begin{equation}\label{eq:mmse_bound}
	\lambda_{r} \, g(\snrdl) \le D_{\rm mmse}\le \lambda_{1} \, g(\snrdl),
\end{equation}
where we have defined  $g(\snrdl) = \trace \left( \mathbf{I} -  \Gm + \Gm ( \mathbf{I}+\Gm)^{-1} \Gm \right) $, to explicitly denote the dependence of this term on $\snrdl$. Note that this dependence emerges from the dependence of $\Xm^{\rm tr}$ and therefore $\Gm$ on $\snrdl$.
We now show how $g(\cdot)$ behaves for large $\snrdl$.
For a given realization of the training matrix $\Xm^{\rm tr}$, denote the eigenvalues of $\Gm$ by $\mu_{i},\, i=1,\ldots, r$. It follows that
\begin{equation}\label{eq:tr_expansion_0}
	\begin{aligned}
		g(\snrdl) &= r - \sum_i \mu_i +\sum_i \frac{\mu_i^2 }{\mu_i + 1} =  r - \sum_{i=1}^r \frac{\mu_i }{\mu_i + 1}
	\end{aligned}
\end{equation} 
Also note that the training matrix can be written as $\Xm^{\rm tr} = \sqrt{\snrdl}\Xm_0^{\rm tr}$, where $\Xm_0^{\rm tr}$ is randomly generated and independent from $\snrdl$. From this and the definition \eqref{eq:G_def_0}, we have $\mu_i = \Theta (\snrdl)$ for all $\mu_i \neq 0$. Using this and \eqref{eq:tr_expansion_0}, we deduce that if $\Gm$ is full-rank ($\mu_i\neq 0$ for all $i$) then $g(\snrdl) = \Theta (\snrdl^{-1})$ and using \eqref{eq:mmse_bound}  we have $ D_{\rm mmse} = \Theta (\snrdl^{-1})$. Conversely, if $\Gm$ has at least one zero eigenvalue ($ \mu_i = 0$ for some $i$) then from \eqref{eq:tr_expansion_0} we have $ g(\snrdl) >1$ for all $\snrdl$ and from \eqref{eq:tr_bounds} we have $ D_{\rm mmse} = \Theta(1)$.

Now, the rank of $\Gm$ depends on the specific realization of $\Xm^{\rm tr}$. When $\beta_{\rm tr}\ge r$ and $\Xm^{\rm tr}$ consists of Gaussian isotropic pilot vectors, $\Gm$ is full-rank with probability one because of the following. The product $\Um_h^\herm \Xm^{{\rm tr}}$ consists of $\beta_{\rm tr} $ independent Gaussian columns, each of dimension $r$. The event that these vectors span a space of dimension less than $r$ has probability zero. Therefore, $\Um_h^\herm \Xm^{\rm tr}  \Xm^{{\rm tr}\, \herm }\Um_h$ has rank $r$ with probability one, and since $\Lambdam_h^{1/2}$ has positive diagonal elements, by definition \eqref{eq:G_def_0} $\Gm$ also has rank $r$ with probability one and $\mu_i \neq 0$ for all $i$. Conversely, if $\beta_{\rm tr}< r$, $\Gm$ has rank at most equal to $\beta_{\rm tr}$ for any realization of the training matrix, leading to $\mu_i= 0$ for some $i$. This results in $D_{\rm mmse}=\Theta (1)$. In short, we have proved $	D_{\rm mmse} = \Theta (\snrdl^{-1})$ for $\beta_{\rm tr} \ge r$ and $	D_{\rm mmse} = \Theta (1)$ for $\beta_{\rm tr}< r$ with probability one over the realizations of $\Xm^{\rm tr}$. In addition, Lemma \ref{lem:rate_distortion} states that only errors $D\ge D_{\rm mmse}$ are achievable. It follows that, if $\beta_{\rm tr}< r$, then the minimum achievable error in estimating the CSIT behaves as $\Theta (1)$. This completes the first part of the proof.

\noindent\textbf{Part II.} To prove the second part, first note that if $\beta_{\rm tr} \ge r$, then the covariance of the MMSE channel estimate $\uv$ at the UE, given as
\begin{equation}\label{eq:Sigmau}
	\Sigmam^u = \Sigmam^h \Xm^{\rm tr}\left(  \Xm^{{\rm tr}\, \herm} \Sigmam^h \Xm^{{\rm tr}} + \mathbf{I} \right)^{-1}\Xm^{{\rm tr}\, \herm} \Sigmam^h - \muv \muv^\herm
\end{equation}
has rank $r$ with probability one over the realizations of $\Xm^{\rm tr}$. Without loss of generality assume the eigenvalues of $\Sigmam^u$ to be ordered as $\lambda_{1}^u \ge \ldots \ge \lambda_{r}^u > 0  $. Next, consider the remote rate-distortion function in Lemma \ref{lem:rate_distortion}, given as
\begin{equation}
	R_{\hv}^r (D) = \sum_{\ell=1}^{MN}\left[ \log \frac{\lambda_{\ell}^{u}}{\gamma} \right]_+,~\text{for}~ D\ge D_{\rm mmse}
\end{equation}
where $\gamma$ is chosen such that $\sum_{\ell=1}^{MN}\min \{ \gamma, \lambda_{\ell}^{u} \} = D-D_{\rm mmse}.$ Consider an interval of error values $D$ for which $D - D_{\rm mmse}<\varepsilon$ for some $\varepsilon$. For sufficiently small $\varepsilon$ we have $\gamma = (D-D_{\rm mmse})/r$ and the remote rate-distortion function is given by 
\begin{equation}\label{eq:R_formula}
	R_{\hv}^r (D) =f(r) - r \log (D-D_{\rm mmse}), ~~\text{for } D-D_{\rm mmse}<\varepsilon,
\end{equation}
where $f(r)= \sum_{\ell = 1}^{r} \log \lambda_{ \ell}^u  + r\log r$ is a value independent of $D$. This implies that for all rates $R>R_{\varepsilon}\triangleq f(r)-r\log \varepsilon$ we can write the remote \textit{distortion-rate} function as 
\begin{equation}\label{eq:dist_rate_func}
	D_{\hv}^r (R) = 2^{\frac{f(r)-R}{r}} + D_{\rm mmse}.
\end{equation}
Now let $R=\beta_{\rm fb} C^{\rm ul}$. Replacing the MIMO-MAC capacity formula $C^{\rm ul} = \log (1+ M\kappa \snrdl)$, we notice that there exists some $\snrdl^\varepsilon$ such that $\beta_{\rm fb} \log (1+M\kappa \snrdl)>R_{\varepsilon}$ for all $\snrdl>\snrdl^{\varepsilon}$. Therefore we can write 
\begin{equation}\label{eq:my_eqqq}
		\scalebox{0.95}{$\log ( D_{\hv}^r (\beta_{\rm fb} C^{\rm ul})- D_{\rm mmse}) = \log r + \sum_{\ell=1}^r \log \lambda_{ \ell}^u /r- \frac{\beta_{\rm fb}}{r}\log (1+ M\kappa \snrdl),~\text{for } \snrdl > \snrdl^{\varepsilon}.$}
\end{equation}
From \eqref{eq:Sigmau} we have that the non-zero eigenvalues of $\Sigmam^u$ scale as $\Theta (1)$, i.e. $\lambda_{\ell}^u = \Theta(1),\, \ell=1,\ldots,r$. Therefore, the right-hand-side of \eqref{eq:my_eqqq} behaves as $\Theta(\log (\snrdl^{-\beta_{\rm fb}/r}))$ in $\snrdl$. It follows that for $\beta_{\rm tr}\ge r$ we have
\begin{equation}\label{eq:my_eq_2}
	\begin{aligned}
		D_{\hv}^r (\beta_{\rm fb} C^{\rm ul}) = D_{\rm mmse} +\Theta(\snrdl^{-\beta_{\rm fb}/r})
		& =\Theta(\snrdl^{-1})+\Theta(\snrdl^{-\beta_{\rm fb}/r})\\
		& =\Theta(\snrdl^{-\min(\beta_{\rm fb}/r,1)}) 
	\end{aligned}
\end{equation}
Finally, from the source-channel separation with distortion theorem, we can achieve a CSIT 
estimation error of $D > D_{\hv}^r ( \beta_{\rm fb} C^{\rm ul} )$ if and only if we use the UL channel over $\beta_{\rm fb}$ feedback dimensions (see Section \ref{sec:RD_LB}), which combined with \eqref{eq:my_eq_2} shows that when $\beta_{\rm tr}\ge r$, we can achieve an error decay of $\Theta(\snrdl^{-\min(\beta_{\rm fb}/r,1)}) $ with probability one over the realizations of $\Xm^{\rm tr}$ with a feedback dimension of $\beta_{\rm fb}$.
Combining Parts I and II of the proof, we have that rate-distortion feedback achieves a CSIT estimation error of $\bE[\Vert \hv - \widehat{\hv} \Vert^2]=\Theta(\snrdl^{-\alpha_{\rm rd}})$, where $\alpha_{\rm rd}= \min (\beta_{\rm fb}/r,1)\mathbf{1}\{\beta_{\rm tr} \geq r \}$.\hfill $\blacksquare$
\begin{remark}\label{remark:alpha_ECSQ}
	We can achieve the same QSE with ECSQ. To see this, note that from \eqref{eq:ECSQ_RD} the remote rate-distortion function with ECSQ for sufficiently small distortion values can be expressed as
	\begin{equation}\label{eq:R_formula_ECSQ}
		R_{\hv}^r (D) =\sum_{\ell = 1}^{r} \log \lambda_{ \ell}^u  + r\log r + 1.508r - r \log (D-D_{\rm mmse}).
	\end{equation}
	Comparing \eqref{eq:R_formula_ECSQ} to \eqref{eq:R_formula}, we note that the same steps leading to equations \eqref{eq:dist_rate_func}-\eqref{eq:my_eq_2} can be repeated by modifying the function $f(\cdot)$ in \eqref{eq:R_formula} to $f(r) = \sum_{\ell = 1}^{r} \log \lambda_{ \ell}^u  + r\log r + 1.508r$. Since the added term $1.508r$ does not depend on SNR, it appears as a constant in the distortion-rate function and the high-SNR error behavior with ECSQ when remains the same. 
\end{remark}

\section{Proof of Theorem \ref{thm:af_distortion_modified}}\label{app:af_thm_proof_modified}

	The channel estimation MMSE given the feedback signal in \eqref{eq:af_bs_signal} can be written as
	\begin{equation}\label{eq:af_distortion_modified}
		\bE [\Vert \hv - \widehat{\hv} \Vert^2] = \trace \left ( \Sigmam^h - \Sigmam^{h}\Xm^{\rm tr}\Psim\Rm_{y,\,{\rm af}}^{-1}\Psim^\herm \Xm^{{\rm tr}\, \herm} \Sigmam^{h} \right) , 
	\end{equation}
	where $\Sigmam^h=\bE[\hv \hv^\herm] $,
	$\bE[\hv \yv^{{\rm af}}] =  \Sigmam^{h}\Xm^{\rm tr}\Psim ,$ and 
	\begin{equation}\label{eq:y_af_cov}
		\Rm_{y,\,{\rm af}} = \Psim^\herm \Xm^{{\rm tr}\, \herm} \Sigmam^h \Xm^{{\rm tr}}\Psim +\Psim^\herm \Psim +\mathbf{I},
	\end{equation}
	Consider the eigen-decomopsition $\Sigmam^h = \Um_h \Lambdam_h \Um_h^\herm$, where $\Um_h\in \bC^{MN\times r} $ is a tall unitary matrix and $\Lambdam_h = {\rm diag}(\lambdav) \in \bR^{r\times r}$ is a diagonal matrix of positive eigenvalues represented by the vector $\lambdav=[\lambda_1,\ldots,\lambda_r]^\transp$, where we assume $\lambda_{1}\ge \ldots\ge \lambda_{r}>0$. Using this decomposition, the expression in \eqref{eq:y_af_cov}, and applying the Sherman-Morrison-Woodbury matrix identity we have
	\begin{equation}
		\begin{aligned}
			\Rm_{y,\,{\rm af}}^{-1} &= (\Psim^\herm \Psim + \mathbf{I})^{-1} 
			-(\Psim^\herm \Psim + \mathbf{I})^{-1} \Psim^\herm \Xm^{{\rm tr}\, \herm}\Um_h \Lambdam_h^{1/2}\times \\
			&\hspace{-20mm} \left( \mathbf{I} + \Lambdam_h^{1/2}\Um_h^\herm \Xm^{{\rm tr}}\Psim (\Psim^\herm \Psim + \mathbf{I})^{-1} \Psim^\herm \Xm^{{\rm tr}\, \herm }\Um_h \Lambdam_h^{1/2}    \right)^{-1}  \Lambdam_h^{1/2} \Um_h^\herm \Xm^{\rm tr} \Psim (\Psim^\herm \Psim +\mathbf{I})^{-1}.
		\end{aligned}
	\end{equation}
It follows that the second term appearing within the ${\rm Tr}(\cdot)$ in \eqref{eq:af_distortion_modified} can be written as
	\begin{equation}\label{eq:eqqq0}
		\begin{aligned}
			\Sigmam^{h}\Xm^{\rm tr}\Psim\Rm_{y,\,{\rm af}}^{-1}\Psim^\herm \Xm^{{\rm tr}\, \herm} \Sigmam^{h} &=\Um_h \Lambdam_h^{1/2}\Gm \Lambdam_h^{1/2}\Um_h^\herm 
			- \Um_h\Lambdam_h^{1/2}\Gm \left( \mathbf{I} + \Gm \right)^{-1}\Gm\Lambdam_h^{1/2}\Um_h^\herm,
		\end{aligned}
	\end{equation}
	where we define
	\begin{equation}\label{eq:G_def}
		\begin{aligned}
			\Gm &\triangleq \Lambdam_h^{1/2}\Um_h^\herm \Xm^{\rm tr}\Psim (\Psim^\herm \Psim + \mathbf{I})^{-1} \Psim^\herm \Xm^{{\rm tr}\, \herm} \Um_h\Lambdam_h^{1/2}\\
		\end{aligned}
	\end{equation}
	Plugging \eqref{eq:eqqq0} into \eqref{eq:af_distortion_modified} we have
	\begin{equation}\label{eq:distor_1}
		\bE [\Vert \hv - \widehat{\hv} \Vert^2]  = \trace \left( \Lambdam_h \left(\mathbf{I} -  \Gm + \Gm ( \mathbf{I}+\Gm)^{-1} \Gm \right)\right). 
	\end{equation}
	This formula is exactly the same as \eqref{eq:Dmmse_2} except for the definition of $\Gm$. Therefore the same trace inequality as in \eqref{eq:mmse_bound} holds here for the CSIT estimation error, i.e. we have 
	\begin{equation}\label{eq:tr_bounds}
		\lambda_{r} \, g(\snrdl) \le \bE [\Vert \hv - \widehat{\hv} \Vert^2]  \le \lambda_{1} \, g(\snrdl),
	\end{equation}
	where $g(\snrdl) = \trace \left( \mathbf{I} -  \Gm + \Gm ( \mathbf{I}+\Gm)^{-1} \Gm \right) $. We now show how $g(\cdot)$ behaves for large $\snrdl$. For a given realization of $\Xm^{\rm tr}$, denote the eigenvalues of $\Gm$ by $\mu_{i},\, i=1,\ldots, r$. We can write
	\begin{equation}\label{eq:tr_expansion}
		\begin{aligned}
			g(\snrdl) &= r - \sum_i \mu_i +\sum_i \frac{\mu_i^2 }{\mu_i + 1} =  r - \sum_i \frac{\mu_i }{\mu_i + 1}
		\end{aligned}
	\end{equation} 
	From the definition in \eqref{eq:G_def}, the constituents of $\Gm$ depend on $\snrdl$ as follows:
	\begin{itemize}
		\item[(a)] We can represent the training matrix as $\Xm^{\rm tr} = \sqrt{\snrdl}\Xm_0^{\rm tr}$, where $\Xm_0^{\rm tr}$ is generated randomly independent from $\snrdl$. Hence,  
		the non-identically zero elements of $\Xm^{\rm tr}$ scale with $\snrdl$ as $\Theta (\sqrt{\snrdl})$ with probability 1.
		\item[(b)] From constraint \eqref{eq:psi_vecs}, we have that each column of $\Psim$ can be written as $\psiv_i = \sqrt{a_i} \phiv_i$, where $a_i = \frac{M\kappa \snrdl }{c_i \snrdl +1}$ with $c_i = \phiv_i^\herm \Xm_0^{{\rm tr}\, \herm }\Sigmam^h \Xm^{\rm tr}_0 \phiv_i$. Here $\phiv_i, \, i\in \{1,\ldots,\beta_{\rm fb}\}$ are a set of unit-norm vectors that contain a subset of $\min(\beta_{\rm tr},\beta_{\rm fb})$ linearly independent vectors, and are independent of $\snrdl$. The existence of this set is guaranteed because $\Xm_0^{{\rm tr}\, \herm }\Sigmam^h \Xm^{\rm tr}_0$ is full-rank with probability one. It follows that
		\begin{equation}
			\begin{aligned}
				\Psim (\Psim^\herm \Psim + \mathbf{I} )^{-1} \Psim^\herm = \Phim (\Phim^\herm \Phim + \Sm )^{-1} \Phim^\herm, 
			\end{aligned}
		\end{equation}
		where $\Phim = [\phiv_1,\ldots,\phiv_{\beta_{\rm fb}}]$ is independent of $\snrdl$ and $\Sm$ is a diagonal matrix whose $(i,i)$ element is given by $S_{i,i} = \frac{c_i}{M\kappa}+ \frac{1}{M\kappa \snrdl}$. The matrix $\Sm$ is the only variable dependent on $\snrdl$, and its diagonal elements scale as $ \Theta (1)$. 
		\item[(c)] The matrices $\Um_h$ and $\Lambdam_h$ are independent of $\snrdl$.
	\end{itemize}
	From these we conclude that the non-zero eigenvalues of $\Gm$ scale as $\Theta (\snrdl )$, i.e. $\mu_i = \Theta (\snrdl)$ for all $\mu_i \neq 0$. Now, the rank of $\Gm$ depends on the specific realization of $\Xm^{\rm tr}$. When $\min(\beta_{\rm tr},\beta_{\rm fb})\ge r$ and $\Xm^{\rm tr}$ contains isotropic Gaussian pilot vectors, $\Gm$ is full-rank with probability one because of the same argument as used in Part I of the proof of Theorem \ref{thm:rate_asymp} and considering the fact that the constituent matrix $\Psim (\Psim^\herm \Psim + \mathbf{I} )^{-1} \Psim^\herm$ is positive semi-definite with rank $\min (\beta_{\rm tr},\beta_{\rm fb})$. In this case we have $\bE [\Vert \hv - \widehat{\hv}\Vert^2 ]=\Theta (\snrdl^{-1})$ and therefore an error of $\Theta (\snrdl^{-1})$ is achievable. Conversely, if $r>\beta_{\rm tr}$ or $r>\beta_{\rm fb}$, $\Gm$ has rank at most $\min (\beta_{\rm tr},\beta_{\rm fb}) < r$ for any design of pilot matrices, leading to $\mu_i= 0$ for some $i$ and from \eqref{eq:tr_bounds}, the error is bounded from below and above by constants, i.e. we have an error of $\Theta (1)$. Therefore analog feedback achieves an error of $\Theta(\snrdl^{-\alpha_{\rm af}}) $ with probability one over the realizations of $\Xm^{\rm tr}$ where $\alpha_{\rm af} = \mathbf{1}\{\min(\beta_{\rm tr},\beta_{\rm fb}) \geq r\}$. This completes the proof.\hfill $\blacksquare$

{\small
	\bibliographystyle{IEEEtran}
	\bibliography{references}
}

%
		
	\end{document}